\documentclass[11 pt]{article}
\usepackage{graphics}
\usepackage{color}
\usepackage{graphicx}
\usepackage{amsmath,bbm,amssymb,mathrsfs, bm}
\usepackage{amsthm}
\usepackage{float} 
\usepackage{booktabs} 
\usepackage{array} 
\usepackage{verbatim} 
\usepackage{float} 
\usepackage{comment}
\usepackage{rotating}
\usepackage{lscape}
\usepackage{enumerate}
\usepackage{fullpage}
\usepackage{hyperref}
\usepackage{setspace}
\usepackage{multirow}
\usepackage{threeparttable}
\usepackage[usenames,dvipsnames]{xcolor}
\usepackage[round]{natbib}
\usepackage{authblk}
\usepackage{enumitem}   

\theoremstyle{definition}
\newtheorem{theorem}{Theorem}
\newtheorem{proposition}{Proposition}
\newtheorem{lemma}{Lemma}
\newtheorem{example}{Example}

\newtheorem{remark}{Remark}
\newtheorem{definition}{Definition}

\def\obs{\text{obs}}
\def\Ytrue{\mathbf{Y}^{\textrm{true}}}
\def\Yimptheta{\mathbf{Y}^{\textrm{imp}}_{\theta}}
\def\Yimpthetazero{\mathbf{Y}^{\textrm{imp}}_{\theta_0}}
\def\Yimpthetai{\mathbf{Y}^{\textrm{imp}}_{\theta,i}}
\def\Yobs{\mathbf{Y}^{\textrm{obs}}}
\def\Tobs{T^{\textrm{obs}}}
\def\Wobs{\mathbf{W}^{\textrm{obs}}}
\def\Dobs{\mathbf{D}^{\textrm{obs}}}
\def\Wrep{\mathbf{W}^{\textrm{rep}}}
\def\Yrep{\mathbf{Y}^{\textrm{rep}}}
\def\Drep{\mathbf{D}^{\textrm{rep}}}
\def\Trep{T^{\textrm{rep}}}
\def\ybarone{\overline{Y}^{\textrm{obs}}(1)}
\def\ybarzero{\overline{Y}^{\textrm{obs}}(0)}
\def\Wrepprime{\mathbf{W}^{\textrm{rep} \prime}}

\title{Leveraging the Fisher randomization test using confidence distributions: inference, combination and fusion learning}
\date{}

\author{Xiaokang Luo, Tirthankar Dasgupta, Minge Xie and Regina Liu}

\affil{Department of Statistics, Rutgers University \\ 110 Frelinghuysen Rd, Piscataway, NJ 08854}

\begin{document}
	\maketitle
	
	\begin{abstract}
The flexibility and wide applicability of the Fisher randomization test (FRT) makes it an attractive tool for assessment of causal effects of interventions from modern-day randomized experiments that are increasing in size and complexity. This paper provides a theoretical inferential framework for FRT by establishing its connection with confidence distributions Such a connection leads to development of (i) an unambiguous procedure for inversion of FRTs to generate confidence intervals with guaranteed coverage, (ii) generic and specific methods to combine FRTs from multiple independent experiments with theoretical guarantees and (iii) new insights on the effect of size of the Monte Carlo sample on the results of FRT. Our developments pertain to finite sample settings but have direct extensions to large samples. Simulations and a case example demonstrate the benefit of these new developments.
\end{abstract}

	\section{Introduction}
	
	Fisher randomization tests (FRT) are flexible tools because they are model free, permit assessment of causal effects of interventions on \textit{any} type of response for \textit{any} assignment mechanism using \textit{any} test statistic, and can be easily extended to model-based inference \citep{Rubin:1980, Rubin:1984}. The tremendous development of computing resources has recently sparked a lot of interest in using FRT to test complex causal hypotheses that can arise from modern-day randomized experiments \citep[e.g.,][]{Hennessy:2016, Athey:2017, Basse_Feller:2018, Basse:2019} in the social, biomedical, educational, behavioral sciences. The work by \citep{Morgan:2012} have shown how randomization tests can be applied to design and analyze randomzied experiments with several pre-treatment covariates. As modern experiments continue to grow in \emph{size} (in terms of number of experimental units, interventions, covariates and as combinations of several independent sub-experiments) and \emph{complexity} (e.g., non-standard randomized assignment mechanisms), the flexibility and wide applicability of FRT make it a promising tool to analyze such experiments.  

However, there are three aspects of FRT that can arguably be made more transparent to make it more appealing to scientists. The first among these is related to the theoretical and implementation aspects of inverting FRTs to generate interval estimators of treatment effects -- interval estimates are typically more appealing than a $p$-value or an acceptance-rejection decision. 
This inversion is done by testing a sequence of sharp null hypotheses of constant treatment effects, and using the curve of the resulting $p$-values. The first original reference of a similar inversion procedure appears in \cite{Pitman:1937}. Whereas proposed procedures and algorithms appear to work well in large sample settings \citep{Garthwaite:1996, Ding:2017}, it is somewhat surprising that the theoretical properties of this inversion procedure, especially in a finite population setting have been scantily discussed in causal inference literature and apparently counter-intuitive simulation results have sometimes been difficult to explain. See, for example, discussion in Sec 7.3 of \cite{Ding:2017} on the intervals for factorial effects obtained in \cite{Dasgupta:2015}. As we shall see in this paper, the discrete nature of the $p$-value statistic poses complexities associated with the inversion procedure in a finite population setting. 

The second, and related aspect is performing meta analysis using FRT. This entails combining results from independently conducted randomized experiments, possibly with different assignment mechanisms, to draw sharper inference on a common treatment effect. Whereas there exist several methods in literature to combine $p$-values from independent tests of hypotheses, obtaining a composite interval with the desired coverage is not straightforward, especially in the finite population case when the $p$-value function is discrete.

The third aspect is computational. The FRT is a computation-intensive procedure, as its classical form involves generating all possible permutations of the observed assignment vector that are consistent with the assignment mechanism. The total number of such permutations in a balanced completely randomized design increases from 252 to $10^{29}$ as the number of units increases from 10 to 100. Procedures like rerandomization \citep{Morgan:2012} involve repeated application of such computations and can be quite intimidating to practitioners.  Researchers studying empirical properties of FRTs using simulations in large $N$ settings are also challenged by such computational requirements. A common way to get around this issue is to generate a sample of all possible permutations, say 1000 or 5000, and use it to obtain a Monte-Carlo estimate of the $p$-value. However, to the best of our knowledge, there does not exist any insights or theoretical results about how large a sample size will guarantee acceptable inferential properties.

In this paper, we attempt to address the three issues mentioned above by providing a new theoretical perspective of FRT using the concept of confidence distributions (CD), which will be formally introduced in Section \ref{ss:CD}. By establishing the $p$-value function of the FRT as an ``approximate'' (defined precisely later) CD function, the paper makes the following contributions: (i) It provides insights into the theoretical properties of the intervals generated by inverting FRTs in finite population settings. In fact it is shown that without precise definitions of $p$-value functions and a carefully designed procedure, inverting the FRT does not necessarily generate an interval with the intended coverage. More surprisingly, it is argued that contrary to the common belief, inversion of FRT using any arbitrary test statistic does not necessarily guarantee an interval for the underlying treatment effect. (ii) By borrowing results from the CD literature, it establishes computationally efficient algorithms for inversion of FRT that generate intervals with the desired coverage. (iii) Establishes procedures for combining inferences from similar and dissimilar experiments using methods for combining CDs (that includes classical methods of combining $p$-values as a special case). (iv) It provides theoretical insights into the effect of the Monte Carlo sample size on the performance of FRT in a finite population setting.



In the following section, we separately introduce the basic notions and concepts of FRT and CD in two subsections. Section 3 creates the bridge between FRT and CD. Section 4 deals with the first objective - developing an unambiguous procedure for inversion of FRTs to generate confidence intervals wth guaranteed coverage. Section 5 provides methods to combine FRTs from independent experiments. Section 6 investigates the effect of size of the Monte Carlo sample on the results of FRT. Section 7 presents a case study and Section 8 contains some concluding remarks.

	\section{Fundamentals} \label{sec:fundamentals}
	
	\subsection{The FRT understood through the potential outcomes model} \label{ss:POT}
	
	Consider a finite population of $N$ experimental units, each of which can be exposed to either a treatment (denoted by 1) or a control (denoted by 0). For unit $i$, let $Y_i(1)$ and $Y_i(0)$ respectively denote the potential outcomes \citep{Neyman:1923, Rubin:1974} under treatment and control. We define the unit-level  causal effect of the treatment on unit $i$ as $\theta_i = Y_i(1) - Y_i(0)$, and the finite-population level average causal effect 
	$$\theta = N^{-1} \sum_{i=1}^N \theta_i = N^{-1} \sum_{i=1}^N Y_i(1) - N^{-1} \sum_{i=1}^N Y_i(0). $$
	
In a randomized design, the $N$ units are assigned to the two treatment groups using a known randomized assignment mechanism. Let ${\mathbf W} = (W_1, \ldots, W_N)^\top$ denote a binary random vector whose $i$th element $W_i$ equals one or zero according as unit $i$ is assigned to treatment or control. The assignment mechanism is defined as the probability distribution of the random vector ${\mathbf W}$ and dictates all inference statements.  In a completely randomized design with $N_1$ and $N_0$ units assigned to treatment and control respectively, where $N_1$ and $N_0$ are predetermined, the assignment mechanism is:
$$	P(W_1 = w_1, \ldots, W_N = w_N) = \left(\frac{N!}{N_0! N_1!}\right)^{-1}  I_{\{\sum_{i=1}^N w_i = N_1\}}. $$	The observed outcome for the $i$th unit is denoted by $ Y_i^\obs = W_i Y_i(1) + (1- W_i) Y_i(0), \ i=1, \ldots, N.$ Thus, only one of the two potential outcomes for each unit is observed and the other is missing.
	
	Consider testing the sharp null hypothesis 
\begin{equation}
H_0^{\theta}: Y_i(1) - Y_i(0) = \theta, \ \mbox{for all} \ i = 1, \ldots, N, \label{eq:H0}
\end{equation} 
that is, all units have an identical treatment effect $\theta$. A special case of this hypothesis is $H_0^0: \theta=0$, Fisher's sharp null hypothesis of no treatment effect on any unit \citep{Fisher:1935, Rubin:1980}. The hypothesis $H_0^{\theta}$ can be tested by considering a suitable test statistic $T$, and comparing its observed value $\Tobs$ with the randomization distribution of $T$ under the null hypothesis. This randomization distribution is generated by imputing the missing outcomes under $H_0^{\theta}$ and repeatedly generating values of $T$ by drawing from the known probability distribution of the assignment vector ${\mathbf W}$. The $p$-value is the tail probability measuring the extremeness of the test statistic with respect to its randomization distribution. Rejection of $H_0^{\theta}$ if the $p$ value is less than or equal to $\alpha \in (0,1)$ leads to a test procedure with level $\alpha$, i.e., the probability of Type-I error not exceeding $\alpha$. The beauty of this procedure is, it can be tested with any reasonable test statistic that is capable of summarizing the difference between the treatment and control groups. 

By varying $\theta$ and testing a set of sharp null hypotheses $H_0^{\theta}$, it is possible to obtain a ``$p$-value function'' of $\theta$, which is a step-value function. This step function can be inverted to generate an interval estimator for the true additive effect $\theta$. As we shall see in Section \ref{sec:connection}, most of the subsequent developments will be based on this $p$-value function and its variants. A toy example presented in the supplementary material demonstrates each step involved in conducting a randomization test, generating a $p$-value function and inverting it to obtain an interval for $\theta$.

\subsection{A Brief Overview of Confidence Distributions and Confidence Curves} \label{ss:CD}
	
The idea of a {\it confidence distribution} (CD) is to use {\it a sample-dependent distribution function} defined on the parameter space to estimate a fixed but unknown (scalar/vector) parameter \citep{cox1958, Efron1993, Efron1998, Xie2013a, Schweder2016}. Such a practice elevates a point (point estimator using the single value of a sample statistic) and two points (confidence interval using a lower and an upper limit) to a full function that can be used to draw inference on the parameter of interest. Similar to a Bayesian posterior, a CD contains much more inferential information than the classical point and interval estimators. 
	
For ease of illustration, consider the simple case of a scalar parameter $\theta \in \Theta$ with sample data ${\mathbf Y}_n  = (Y_1, \ldots, Y_n) \in {\cal Y}$.  A function $H_n(\cdot)\equiv H(\cdot, {\mathbf Y}_n)$ on $\Theta \times {\cal Y}$ is called a {\it confidence distribuition (CD) function} for $\theta$, if (i) given ${\mathbf Y}_n$, $H_n(\cdot)$ is a cumulative distribution function on $\Theta$; and (ii) at the true parameter value $\theta=\theta_0$, $H_n(\theta_0) = H(\theta_0, {\mathbf Y}_n)$, as a function of the sample ${\mathbf Y}_n$, follows a Uniform[0,1] distribution \citep{Schweder2002, Singh2005}.  In other words, (i) requires that a CD is a sample-dependent distribution function on $\Theta$. Requirement (ii) ensures that the CD function can be used to obtain confidence intervals and test hypotheses. For example, by (ii), $(-\infty, H_n^{-1}(\alpha))$ is a $100 (1-\alpha)\%$ confidence interval for $\theta$, and $H_n(b)$ provides a $p$-value function for testing the hypothesis $\Omega_0: \theta \leq b$ versus $\Omega_1: \theta > b$. This shows that a one-sided $p$-value function is a special case of a CD.  Corresponding to a CD function $H_n(\theta)$, one can obtain a \emph{confidence curve} (CV) 
$$ CV(\theta) = 2 \min\{H_n(\theta), 1-H_n(\theta)\},$$ 
which can also be used to draw similar inferences \citep{Birnbaum:1961}.


Due to the discrete nature of the FRT in which the $p$-value is a step function as in Figure  \ref{fig:Trep_hist}(supplementary material), the following definition will be useful for this paper:
\begin{definition}[\textbf{Upper and Lower CDs}] A function $H_n^{\text{L}}(\cdot) = H^{\text{L}}(\cdot, {\mathbf Y}_n)$ mapping $\Theta \times {\cal Y}$ to [0,1] is said to be a lower CD for a parameter $\theta$ if at the true parameter value $\theta = \theta_0$, $H_n(\theta_0) \equiv H^{\text{L}}(\theta_0, \mathbf{Y}_n)$, as a function of the sample $\mathbf{Y}_n$ is stochastically larger than a Uniform[0,1] random variable, that is,
\begin{equation}
P \left[H^{\text{L}}(\theta_0, \mathbf{Y}_n) \le \alpha  \right] \le \alpha \quad \mbox{for all} \quad \alpha \in (0,1). \label{eq:lowerCD}
\end{equation}
\noindent An upper CD  $H_n^{\text{U}}(\cdot) = H^{\text{U}}(\cdot, \mathbf{Y}_n)$ for parameter $\theta$ can be defined similarly but with (\ref{eq:lowerCD}) replaced by $P \left[H^{\text{U}}(\theta_0, \mathbf{Y}_n) \le \alpha  \right] \ge \alpha$ for all $\alpha \in (0,1)$.  \label{def:upperlowerCD}
\end{definition}

\section{Connecting FRT to CD through the $p$-value function} \label{sec:connection}
	
We note that both FRT and CD historically have an implicit ``fiducial'' flavor, although in recent developments \citep{Schweder2016, Xie2013a}, the concept of CD has been developed without any fiducial interpretation or reasoning. Some researchers consider a CD as ``a frequentist analog of a Bayesian posterior''\citep{Schweder2003b}. On the other hand, \cite{Rubin:1984} provided the following Bayesian justification of the FRT: it gives the posterior predictive distribution of the estimand of interest under a model of constant treatment effects and fixed units with fixed responses. These connections motivate us to understand the properties of FRT better by connecting it to CD and exploiting recent results on CD. At the same time, this connection is non-trivial because the theory of CD primarily revolves around parametric models, whereas FRT is essentially a model-free procedure. At the same time, the discrete nature of the distribution of the $p$-value in FRT also adds to this complication.

We first extend the notion of the $p$-value for the FRT to a $p$-value function along the lines of that introduced in Section \ref{ss:CD}. To do this, we start with a more careful handling of the notations involved. Let $\Ytrue$ denote the true matrix of potential outcomes and $\Yimptheta$ the imputed matrix under the null hypothesis $H_0^{\theta}$. Let $\Wobs$ denote the $N \times 1$ observed assigned vector and $\Yobs$ the $N \times 1$ observed vector of responses. Then the observed data from the experiment can be denoted by $\Dobs = (\Yobs, \Wobs)$. Also, let $\Wrep$ denote any repeated draw from the distribution of ${\mathbf W}$ while generating the randomization distribution of $T$. Such a repeated draw generates repeated data $\Drep = (\Yrep, \Wrep)$, where $\Yrep$ is a random sample from $\Yimptheta$. 
	
	Let $T$ be any test statistic and $\Trep$ denote the discrete random variable having the randomization distribution of $T$. The the distribution of $\Trep$ depends on the imputed potential outcomes matrix $\Yimptheta$ and $\Wrep$. Consequently, we can write 
	\begin{equation}
	\Trep = T(\Drep) = T(\Yimptheta, \Wrep). \label{eq:Trep}
	\end{equation}
	Finally, note that the observed value of the test statistic $\Tobs$ depends on $\Dobs$, and consequently on $\Ytrue$ and $\Wobs$, allowing us to write
	\begin{equation}
	\Tobs = T(\Dobs) = T(\Ytrue, \Wobs). \label{eq:Tobs}
	\end{equation}

\subsection{$p$-value functions for one-sided alternatives of the sharp null} \label{ss:one-sided}

Whereas the sharp null hypothesis has been widely discussed in literature, the alternative hypothesis against which the sharp null is tested has seldom been mentioned. In this paper, we will keep our alternatives restricted to the class of sharp-nulls to make the interval estimation problem readily interpretable. A violation of the sharp null can be one-sided or two-sided. Below, we define $p$-value functions for one-sided alternative hypotheses.

\begin{definition} \label{def:p_plus}
Consider the one sided alternative 
\begin{equation}
H_1^{\theta^+}: Y_i(1) - Y_i(0) = \phi (> \theta), \label{eq:H1+}
\end{equation}
for all $i = 1, \ldots, N$. Assuming that larger values of the test statistic $T$ indicate departure from the sharp null in favor of $H_1^{\theta^+}$, we define the following $p$-value functions for testing $H_0^{\theta}$ against alternatives $H_1^{\theta^+}$ as: 
\begin{eqnarray}
p^{L+}(\Dobs, \theta) &=& P(\Trep \ge \Tobs) = P \left( T(\Yimptheta, \Wrep) \ge T(\Dobs) \right)  \label{eq:p_right_L} \\
p^{U+}(\Dobs,\theta) &=& P(\Trep > \Tobs) = P \left( T(\Yimptheta, \Wrep) > T(\Dobs) \right).  \label{eq:p_right_U}
\end{eqnarray} 
\end{definition}

\begin{definition} \label{def:p_minus}
\begin{equation}
H_1^{\theta^-}: Y_i(1) - Y_i(0) = \psi (< \theta), \label{eq:H1-}
\end{equation} 
for all $i = 1, \ldots, N$. Assuming that smaller values of the test statistic $T$ indicate departure from the sharp null in favor of $H_1^{\theta^-}$, we define the $p$-value function for testing $H_0^{\theta}$ against alternatives $H_1^{\theta^-}$ as 
\begin{eqnarray}
p^{L-}(\Dobs, \theta) &=& P(\Trep \le \Tobs)  = P \left( T(\Yimptheta, \Wrep) \le T(\Dobs) \right). \label{eq:p_left_L} \\
p^{U-}(\Dobs, \theta) &=& P(\Trep < \Tobs)  = P \left( T(\Yimptheta, \Wrep) < T(\Dobs) \right). \label{eq:p_left_U}
\end{eqnarray} 
\end{definition}

Note that the $p$-value functions defined in (\ref{eq:p_right_L})-(\ref{eq:p_left_U}) are random variables because of their dependence on $\Dobs$. However, conditional on $\Dobs$, they are functions of $\theta$.

\begin{proposition} \label{prop:FRTlowerCD}
For any test statistic $T$, the $p$-value functions defined in (\ref{eq:p_right_L})-(\ref{eq:p_left_U}) satisfy the following properties:
\begin{itemize}
\item[(a)] Both $p^{L+}(\Dobs, \theta)$ in (\ref{eq:p_right_L}) and $p^{L-}(\Dobs, \theta)$ in (\ref{eq:p_left_L}) are lower CDs as per Definition \ref{def:upperlowerCD}, which means they both stochastically dominate the Uniform[0,1] random variable at the true value $\theta_0$ of $\theta$ and satisfy 
$$P \left(p^{L+}(\Dobs, \theta_0) \le \alpha \right) \le \alpha, \quad \mbox{and} \quad P \left(p^{L-} (\Dobs, \theta_0) \le \alpha \right) \le \alpha,$$ for $\alpha \in (0,1)$, where $\theta_0$ is the true value of $\theta$. 
\medskip
\item[(b)] Both $p^{U+}(\Dobs, \theta)$ in (\ref{eq:p_right_U}) and $p^{U-}(\Dobs, \theta)$ in (\ref{eq:p_left_U}) are upper CDs in the sense that they are both stochastically dominated by the Uniform[0,1] random variable and satisfy 
$$P \left(p^{U+}(\Dobs, \theta_0) \le \alpha \right) \ge \alpha, \quad \mbox{and} \quad P \left(p^{U-} (\Dobs, \theta_0) \le \alpha \right) \ge \alpha \quad ,$$ for $\alpha \in (0,1)$, where $\theta_0$ is the true value of $\theta$. 
\item[(c)] Let $T_{(1)} < T_{(2)} < \cdots < T_{(m)}$ be the $m$ unique ordered values of $T$ for $\theta = \theta_0$ and $\gamma_i=P\left(T(\Ytrue, \mathbf{W})=T_{(i)}\right)>0$ for $i=1,2,\cdots, m$. Then, for any $\alpha \in (0,1)$, 
\begin{eqnarray}
P \left( p^{L+}(\Dobs,\theta_0) \le \alpha \right) \ge \alpha - \gamma^*, \quad
P \left( p^{L-}(\Dobs,\theta_0) \le \alpha \right) \ge \alpha - \gamma^*, \nonumber \\
P \left( p^{U+}(\Dobs,\theta_0) \le \alpha \right) \le \alpha + \gamma^*, \quad
P \left( p^{U-}(\Dobs,\theta_0) \le \alpha \right) \le \alpha + \gamma^* \label{eq:discrepancy}
\label{eq: error inequality}
\end{eqnarray}
where $\gamma^* = \max\{\gamma_1,\gamma_2,\cdots,\gamma_m\}$.
\end{itemize}
\end{proposition}

\noindent\textbf{Implications of Proposition \ref{prop:FRTlowerCD} and some remarks}

\begin{enumerate}
\item Consider testing the sharp null hypothesis (\ref{eq:H0}) against one-sided alternatives (\ref{eq:H1+}) or (\ref{eq:H1-}) using a test statistic whose large or small values indicate departure from the null in favor of (\ref{eq:H1+}) or (\ref{eq:H1-}) respectively. By part (a) of Proposition \ref{prop:FRTlowerCD}, the test procedure that rejects the sharp null if the observed value of $p^{L+}(\Dobs, \theta) \le \alpha$ is \emph{valid} in the sense that it the probability of Type-I error does not exceed  $\alpha$. However, by part(b), the rejection rule $p^{U+}(\Dobs, \theta) \le \alpha$ is not valid. Similarly, use of $p^{L-}(\Dobs, \theta)$ for the one-sided alternative (\ref{eq:H1-}) leads to a valid test, while use of $p^{U+}(\Dobs, \theta)$ does not.

\item Equations (\ref{eq:discrepancy}) provide upper bounds for the discrepancies between the empirical CDFs of  the four $p$-value functions given by (\ref{eq:p_right_L})-(\ref{eq:p_left_U}) from the CDF of a Uniform[0,1] variable. This is illustrated in the left panel of Figure \ref{fig:gammabound} that shows a partial plot of the empirical CDF of $p^{L+}(\Dobs,\theta_0)$ (based on the toy example in the supplementary material) against the Uniform[0,1] CDF. In this example, $\gamma$'s take only two values: 1/252 and 2/252, and the maximum discrepancy between the CDFs is seen as $\max \gamma_m = 2/252$ on three occasions. A similar plot of the empirical CDF of $p^{L-}(\Dobs,\theta_0)$ against the Uniform[0,1] CDF is shown in the right panel.


	\begin{figure}[ht]
		\centering
		\caption{\small Empirical CDFs of $p^{L+}(\Dobs,\theta_0)$ (left panel) and $p^{L-}(\Dobs,\theta_0)$ (right panel) versus Uniform[0,1] CDF showing the discrepancy $\gamma$'s} \label{fig:gammabound}
		\centering{\includegraphics[width=15 cm]{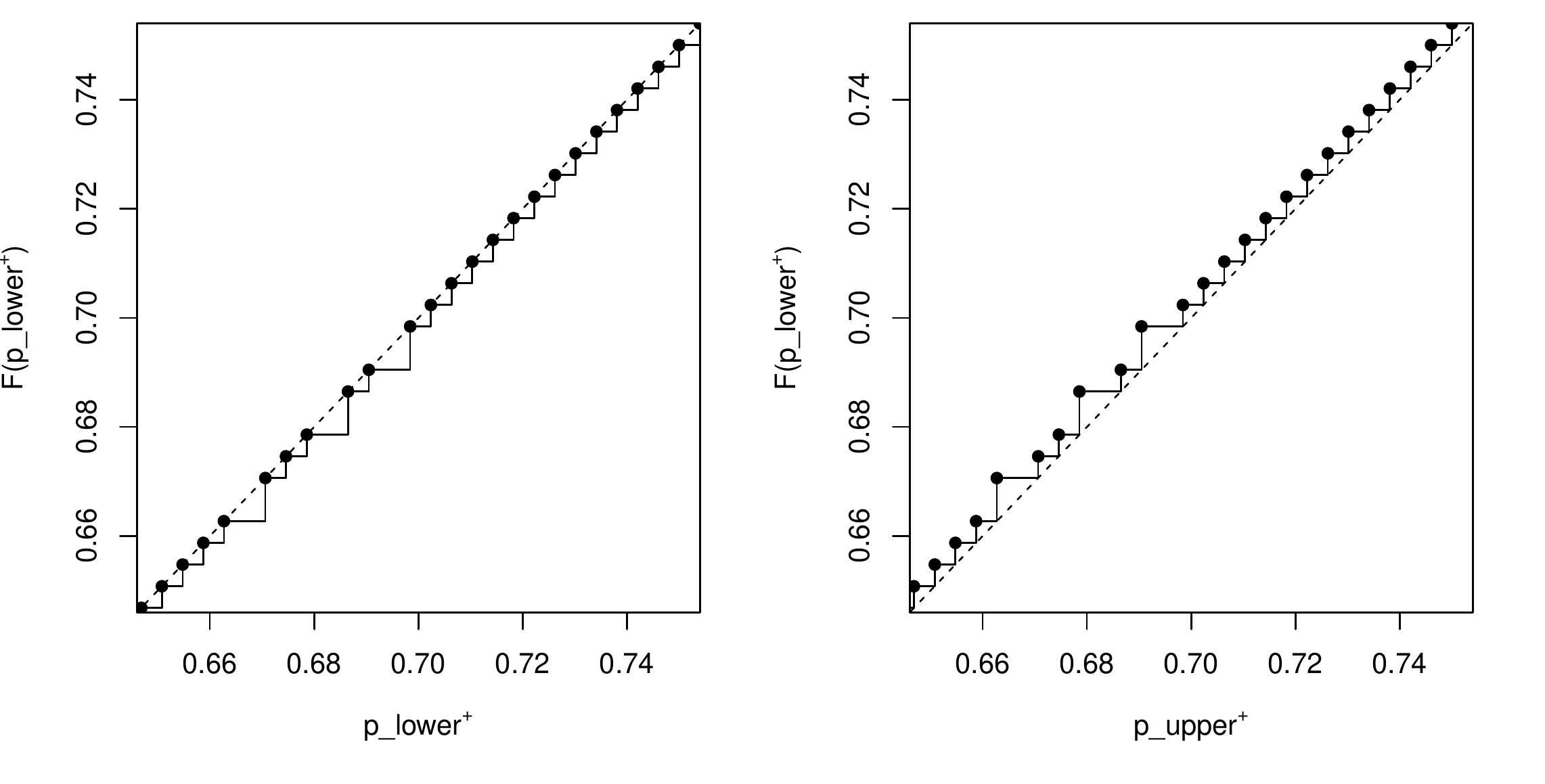}}
		\singlespace
	\end{figure} 

\end{enumerate}

\subsection{Two-sided alternatives} \label{ss:two-sided}

We now consider testing the sharp null $H_0^{\theta}$ against a two-sided alternative hypotheses
\begin{equation} 
H_1^{\theta_{\pm}}: Y_i(1) - Y_i(0)= \eta \  (\ne \theta), \ \mbox{for all} \ i=1, \ldots, N. \label{eq:H1both}
\end{equation} 

\begin{definition} \label{def:p_plusminus}
The $p$-value function for testing $H_0^{\theta}$ against alternatives $H_1^{\theta^{\pm}}$ is
\begin{equation}
p^{L}(\Dobs, \theta) =	2 \min \big\{ p^{L+}(\Dobs, \theta), p^{L-}(\Dobs, \theta)\big\}, \label{eq:p_twosided}
\end{equation}
where $p^{L+}(\Dobs, \theta)$ and $ p^{L-}(\Dobs, \theta)$ are defined in (\ref{eq:p_right_L}) and (\ref{eq:p_left_L}) respectively.
\end{definition}

\begin{figure}[ht]
			\centering \caption{$p^{L}(\Dobs,\theta)$ vs $\theta$} \label{fig:pplusminus}
				\includegraphics[width=8cm]{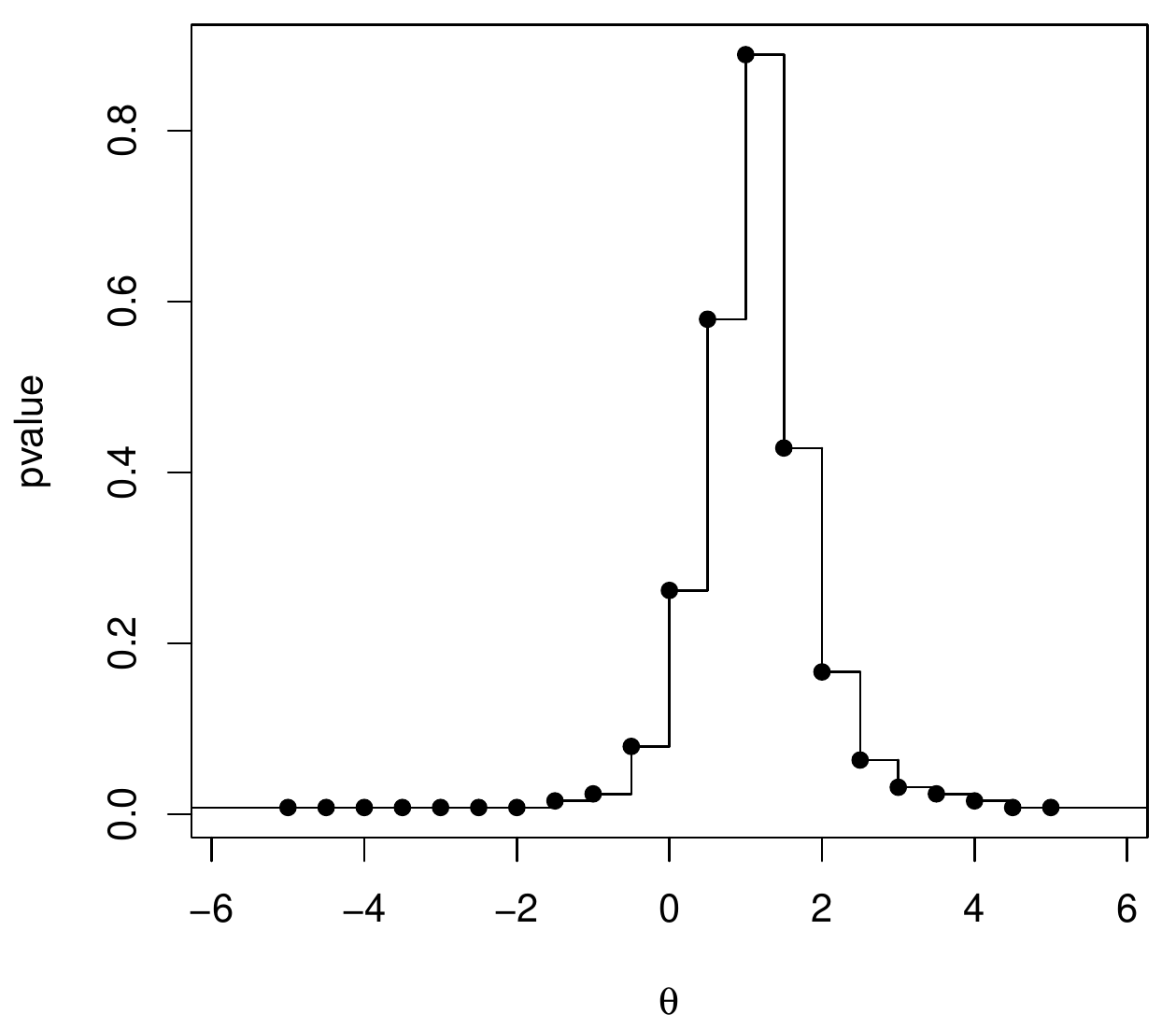} 
		\end{figure}
		
The function $p^{L}(\Dobs, \theta)$ can be considered a discrete version of a CV function. By part (a) of Proposition \ref{prop:FRTlowerCD},  $p^{L+}(\Dobs, \theta)$ and $p^{L-}(\Dobs, \theta)$ stochasticaly dominate a Uniform[0,1] random variable when $\theta = \theta_0$ and thus is a valid $p$-value function to test $H_0^{\theta}$ against $H_1^{\theta_{\pm}}$. Note that, if the $p$-value function for this two-sided testing problem had been constructed along the lines of the CV function introduced in Section \ref{ss:CD} as
 $$2\min\{p^{L+}(\Dobs, \theta),1-p^{L+}(\Dobs, \theta)\}=2\min\Big\{p^{L+}(\Dobs, \theta), p^{U-}(\Dobs,\theta)\Big\},$$  
then it would not have dominated a Uniform[0,1] random variable by part (b) of 	Proposition \ref{prop:FRTlowerCD}.	

Figure \ref{fig:pplusminus} illustrates a $p^{L}(\Dobs, \theta)$ function based on the toy example in the supplementary materials.

\section{Inverting the FRT to obtain confidence intervals} \label{sec:CI}
	
As briefly mentioned in the introductory section, the procedure of inverting FRTs to obtain intervals for treatment effects has been described rather loosely in literature. For example, \cite{Dasgupta:2015} described the procedure for obtaining a $100(1-\alpha)\%$ Fisherian interval as (i) Computing $p$-values $P(\Trep \ge \Tobs)$ (which, in our notation, is $p^{L+}(\Dobs, \theta)$) by testing a sequence of sharp null hypothesis for different values of $\theta$ to create a $p$-value function $p(\theta)$ and (ii) obtaining the lower and upper limits for $\theta$ by inverting the function at $\alpha/2$ and $1-\alpha/2$ respectively. Such a procedure is demonstrated in the lower right panel of Figure \ref{fig:Trep_hist} in the supplementary material. This procedure has two problems. First, it implicitly assumes that $p^{L+}(\Dobs, \theta)$ is a monotonically increasing function of $\theta$. However, such monotonicity is not guaranteed by Proposition  \ref{prop:FRTlowerCD}. Thus, inversion of the $p$-value function will not necessarily produce intervals for the treatment effect. We will soon see that one needs to impose additional conditions on the test statistic to guarantee monotonicity of the $p$-value function.

The second problem arises from the fact that even if we assume that $p^{L+}(\Dobs, \theta)$ is monotonic, part (a) of Proposition  \ref{prop:FRTlowerCD} implies:
\begin{eqnarray*}
P \left(p^{L+} (\Dobs, \theta_0) \le \alpha/2 \right) \le \alpha/2, \ \mbox{but} \ P \left(p^{L+}(\Dobs, \theta_0) \ge 1-\alpha/2 \right) \ge \alpha/2.
\end{eqnarray*}

Therefore, inverting the same function $p^{L+}(\Dobs, \theta)$ does not guarantee that the coverage of the generated interval will be $1-\alpha$. In this following subsections, we address these two issues: monotonicity of the $p$-value functions and the correct way of inverting them to guarantee the right coverage.
	
\subsection{Monotonicity of the $p$-value function} \label{ss:monotonicity}

We first give a counter example that the function $p^{L+} (\Dobs, \theta)$ can be non-monotonic in specific settings. Consider an similar to Example \ref{example:toy1} discussed earlier. 
\begin{example} \label{example:toy3}
Consider the following Table of potential outcomes generated from two lognormal distributions where the true additive treatment effect is 1.
\begin{table} [ht]
\centering
\caption{Potential outcomes and observed data in example 3} \label{tab:toy2}
\begin{tabular}{c|cc|c|c}       
Unit($i$) & $Y_i(0)$ & $Y_i(1)$ & $W_i^{\textrm{obs}}$ & $Y_i^{\textrm{obs}}$ \\ \hline
1  & 0.14 & 1.14  &  1 & 1.14 \\
2 & 1.12  & 2.12 &   1 & 2.12 \\
3  & 0.80 & 1.80 &   0 & 0.80 \\
4  & 1.80 & 2.80 &   1 & 2.80 \\
5  & 0.90 & 1.90 &   0 & 0.90 \\
6  & 0.44 & 1.44 &   0 & 0.44 \\
7  & 1.13 & 2.13 &   1 & 2.13 \\
8  & 0.53 & 1.53 &   0 & 0.53 \\\hline
			\end{tabular}
		\end{table}
Consider two test statistics to test the sharp null $H_0^{\theta}$ defined in (\ref{eq:H0}) against the one-sided alternative (\ref{eq:H1+}). The first is the unnormalized difference of treatment means  $T_1 = \ybarone - \ybarzero$ used in all the previous examples, and the second is the ``studentized'' Fisher-Behren type statistic $T_2 =  \left\{\ybarone - \ybarzero \right\}/ \left\{\sqrt{s_1^2/4 + s_0^2/4} \right\}$ where $s_1^2$ and $s_0^2$ are the sample variances of the observed treatment and control outcomes respectively. By part (b) of Proposition  \ref{prop:FRTlowerCD}, the test procedure that rejects (\ref{eq:H0}) if $p^{L+}(\Dobs, \theta) \le \alpha$ is a valid procedure, irrespective of whichever test statistic is used. However, the plots of $p$-value functions $p^{L+}(\Dobs, \theta)$ constructed with $T_1$ and $T_2$, as shown in Figure \ref{fig:monotonicity}, reveal a different story: the function  $p^{L+}(\Dobs, \theta) \le \alpha$ constructed with statistic $T_2$ is non-monotonic, and hence its inversion does not necessarily generate an interval for every choice of $\alpha \in (0,1)$. However, it is somewhat relieving to note that the function is monotonic in some neighborhood of $\theta_0 = 1$.
\begin{figure}[ht]
		\centering
		\caption{\small $p^{L+}(\Dobs, \theta)$ obtained from the data in Table \ref{tab:toy2} using test statistic $T_1 = \ybarone - \ybarzero$ (left panel) and $T_2 = \left\{\ybarone - \ybarzero \right\}/ \left\{\sqrt{s_1^2/4 + s_0^2/4}\right\}$ (right panel) } \label{fig:monotonicity}
      \vspace{0.1 in}
		\centering{\includegraphics[width=15 cm]{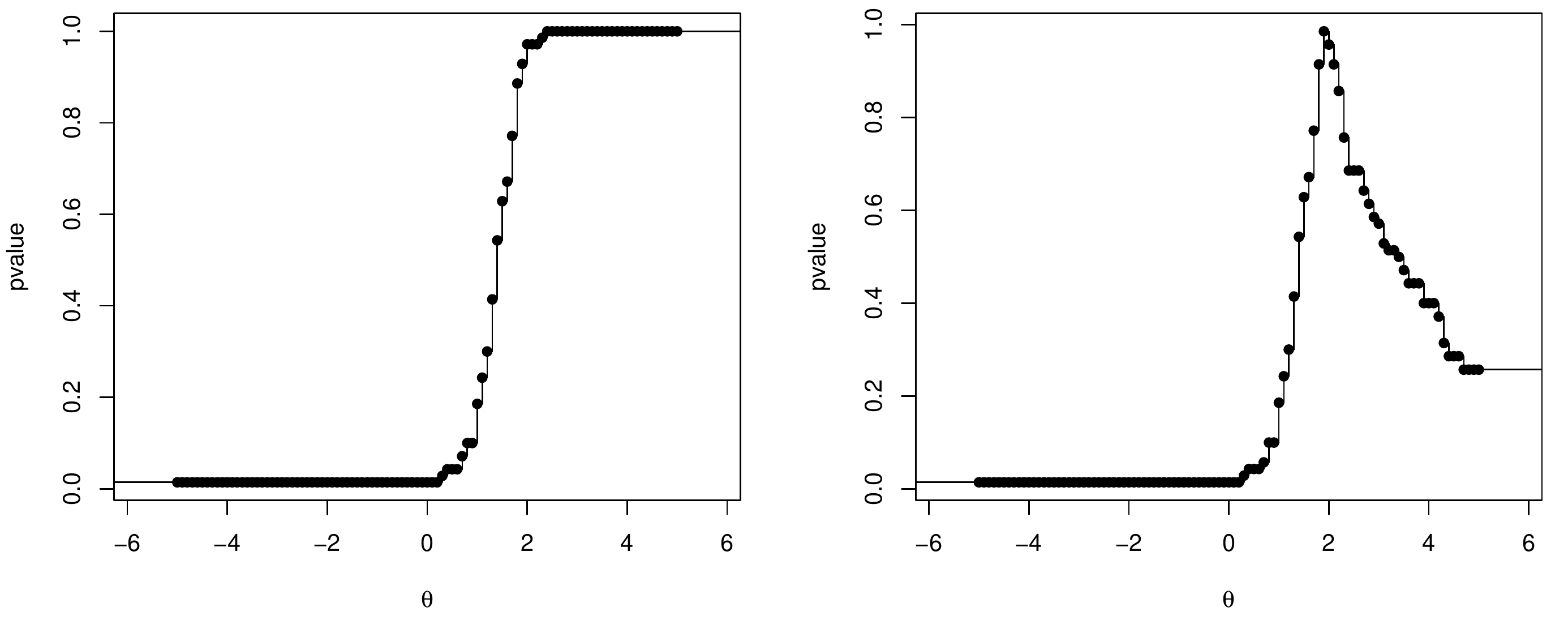}}
		\singlespace
	\end{figure} 
\end{example}

Example \ref{example:toy3} indicates that the behavior of the $p$-value functions defined in (\ref{eq:p_right_L})-(\ref{eq:p_left_U}) depend on the choice of the test statistic $T$. We now provide a result that gives a set of sufficient conditions to guarantee that $p^{L+}(\Dobs, \theta)$ ``behaves'' like a CDF in the sense that is monotonically non-decreasing and right continuous. We first introduce the following definitions along the lines of \cite{Caughey:2017}.

\begin{definition}[\textbf{Ordered vectors of potential outcomes}] \label{def:order_PO}
Two vectors of potential outcomes under treatment $\mathbf{Y}(1) = \left(Y_i(1), \ldots, Y_N(1) \right)$ and $\mathbf{Y}^{\prime}(1) = \left(Y_1^{\prime}(1), \ldots, Y_N^{\prime}(1) \right)$ are ordered as $\mathbf{Y}(1) \le \mathbf{Y}^{\prime}(1)$ if $Y_i(1) \le Y^{\prime}_i(1)$ for all $i=1, \ldots, N$. An order between two vectors of potential outcomes under control, $\mathbf{Y}(0)$ and $\mathbf{Y}^{\prime}(0)$ is similarly defined.
\end{definition}

\cite{Caughey:2017} introduced the notion of an ``effect increasing'' (EI) statistic in the context of testing null hypotheses that are weaker than the sharp null. A definition of EI test statistic is given below.

\begin{definition}[\textbf{Effect increasing (EI) test statistic}] \label{def:EI_statistic}
A test statistic $T(\mathbf{Y},\mathbf{W}) = T(\mathbf{Y}(1), \mathbf{Y}(0), \mathbf{W})$ is said to possess the EI property if it is  non-decreasing in $\mathbf{Y}(1)$ and non-increasing in $\mathbf{Y}(0)$. 
\end{definition}


\noindent Examples of EI statistics include difference in means or Wilcoxon rank sum statistic. On the other hand, the studentized Fisher-Behren type statistic defined in Example \ref{example:toy3} does not satisfy the EI property.

\begin{theorem} \label{thm:monotonicity}

\begin{enumerate}
\item[(a)] If the test statistic $T$ is EI, then the $p$-value function $p^{L+}(\Dobs, \theta)$ defined in (\ref{eq:p_right_L}) is non-decreasing in $\theta$ for fixed $\Dobs$, 
\item[(b)] Further, for fixed $\mathbf{W}$, if $T(\Yimptheta, \mathbf{W})$ is right continuous as a function of $\theta$, and approaches $-\infty$ and $+\infty$ as $\theta \rightarrow -\infty$ and $\theta \rightarrow +\infty$ respectively, then (i) $p^{L+}(\Dobs, \theta) \rightarrow 1$ as $\theta \rightarrow \infty$ and $p^{L+}(\Dobs, \theta) \rightarrow P(\mathbf{W}=\Wobs)$ as $\theta \rightarrow -\infty$, (ii) $p^{L+}(\Dobs, \theta)$ is right continuous in $\theta$.
\end{enumerate}
\end{theorem}

\begin{remark}  \label{remark:pminusmonotonicity}
Results similar to Theorem \ref{thm:monotonicity} also hold for $p^{L-}(\Dobs, \theta)$ defined in (\ref{eq:p_left_L}), which is non-increasing if $T$ is EI.
\end{remark}

\begin{remark}  \label{remark:Caughey-monotonicity}
\cite{Caughey:2017} pointed out the important role of EI statistics in constructing valid tests for null hypothesis that are weaker than the sharp null. Theorem  \ref{thm:monotonicity} establishes that this condition is sufficient for monotonicity of $p$-value functions in FRT.

\end{remark}

\subsection{Algorithm for generating intervals with coverage at least $1-\alpha$} \label{ss:CI_algorithm}

From the foregoing discussion, it is clear that the ``traditional'' approach of inverting just \emph{one} $p$-value function based on an \emph{arbitrary} test statistic does not necessarily led to generation of one or two-sided intervals with the desired coverage. Based on (i) the properties of the $p$-value functions in Proposition \ref{prop:FRTlowerCD}, (ii) desription of valid procedures for testing the sharp null against one or two sided hypotheses in Section \ref{sec:connection}, and (iii) conditions required to guarantee that inversion of $p$-value functions will generate intervals as stated in Theorem \ref{thm:monotonicity}, we now arrive at the following proposition that provides us with a rule to generate confidence intervals with the desired coverage.

\begin{proposition} \label{prop:CI}
Assume that for fixed $\Dobs$, the $p$-value functions $p^{L+}(\Dobs,\theta)$ and $p^{L-}(\Dobs,\theta)$ are (i) respectively non-decreasing and non-increasing and (ii) right continuous functions of $\theta$. 
\begin{enumerate}
\item[(a)] Define $\theta_{\ell}(\alpha) = \sup_\theta\{\theta: p^{L+}(\Dobs,\theta) \le \alpha\}.$ Then the one-sided interval $\left[\theta_{\ell}(\alpha), \infty \right)$ covers the true value of $\theta$ with probability of at least $1-\alpha$.
\item[(b)] Define $\theta_u(\alpha) = \inf_\theta\{\theta: p^{L-}(\Dobs,\theta) \le \alpha\}$. Then the one-sided interval $\left(-\infty, \theta_u(\alpha) \right)$ covers the true value of $\theta$ with probability of at least $1-\alpha$.
\item[(c)] For $0 < \alpha_1, \alpha_2 < 1$ and $\alpha_1 + \alpha_2 = \alpha$, define $\theta_{\ell}(\alpha_1) = \sup_\theta\{\theta: p^{L+}(\Dobs,\theta) \le \alpha_1\}$ and $\theta_{u}(\alpha_2) = \inf_\theta\{\theta: p^{L-}(\Dobs,\theta) \le \alpha_2\}$. Then the two-sided interval $\left[\theta_{\ell}(\alpha_1),\theta_u(\alpha_2) \right)$ covers the true value of $\theta$ with probability of at least $1-\alpha$.
\end{enumerate}
\end{proposition}

Proposition \ref{prop:CI} provides us with methods to determine intervals for the treatment effect with the desired coverage. The most straightforward approach is to obtain the interval $\left[\theta_{\ell}(\alpha/2), \theta_u(\alpha/2) \right]$ where $\theta_{\ell}(\alpha/2)$ and $\theta_u(\alpha/2)$ are obtained by substituting $\alpha_1 = \alpha_2 = \alpha/2$ in part (c) of Proposition \ref{prop:CI}. To obtain $\theta_{\ell}(\alpha/2)$ and $\theta_u(\alpha/2)$, one can respectively solve the equations
\begin{equation}
p^{L+}(\Dobs,\theta) = \alpha/2, \ p^{L-}(\Dobs,\theta) = \alpha/2, \label{eq:solve_two}
\end{equation}
which is equivalent to solving
\begin{equation}
p^{L}(\Dobs,\theta) = \alpha. \label{eq:solve_one}
\end{equation}
Recall that the traditional approach is to obtain the interval by solving
\begin{equation}
p^{L+}(\Dobs,\theta) = \alpha/2, \ p^{L+}(\Dobs,\theta) = 1-\alpha/2. \label{eq:solve_traditional}
\end{equation}


	
The difference in coverage between the intervals generated by the traditional and proposed approach tend to be pronounced for populations of small sizes, and particularly when the number of unique values of the test statistic is small. Example \ref{example:discrete} below demonstrates this difference.

\begin{example} \label{example:discrete}
Consider a table of potential outcomes with $N=15$ in which potential outcomes $(Y(0), Y(1))$ are (0,0) for six units, (1,1) for six units and (2,2) for three units. Thus the true value of $\theta$ is zero. The coverages of the interval for $\theta$ generated by the traditional and proposed approach from this population are 0.897 and 0.961 respectively.
\end{example}
	
\section{CD as a tool for combining FRTs from independent studies} \label{sec:meta_analysis}
	
	Large $N$ studies now frequently arise from aggregation of information from multiple independent sources \citep[e.g.][]{Hemkens:2017} and require strategies for efficient meta-analysis. Several researchers \citep[e.g.][]{Pearl:2016} have emphasized on the importance of development of new methodologies for combining information from multiple sources, stating that the objective of such fusion inference is ``to combine results from many experimental and observational studies, each conducted on a different population and under a different set of conditions in order to synthesize an aggregate measure of targeted effect size that is \textit{better}, in some sense, than any one study in isolation.'' Our research is motivated by the desire to have such developments.
 
There exist several classical methods in literature to combine $p$-values from independent tests of hypotheses, e.g., Fisher's method \citep{fisher1932statistical}, Stouffer's method \citep{stouffer1949american}. See \cite{Marden1991} for a detailed review of these and other methods. However, while it is straightforward to combine $p$-values from multiple independent tests, it is not obvious how to combine the results into a \emph{composite p-value function} from which a composite interval estimator for $\theta$ can be obtained. \cite{Singh2005} and \cite{xie2011confidence} proposed a general recipe to combine CDs, and specifically $p$-value functions, that encompass all the classical methods for combining $p$-values as special cases. Exploiting the connection between FRT and CD developed in the previous two sections, and utilizing the framework of \cite{xie2011confidence}, we propose the following procedure for combining $p$-value functions for testing the same sharp null hypothesis $H_0^{\theta}$ from $M$ independent experiments.

For $i=1, \ldots, M$, let $p_i^{L+}(\theta) = p^{L+}(\Dobs_i, \theta)$ defined by (\ref{eq:p_right_L}) denote one-sided $p$-value functions obtained from $M$ independent randomized experiments. We combine the $p$-value functions $p_1^{L+}(\theta),  \ldots, p_M^{L+}(\theta)$ to obtain a combined $p$-value function 
\begin{equation}
p_c^{L+}(\theta) = G_c \left( g_c \left (p_1^{L+}(\theta),  \ldots, p_M^{L+}(\theta) \right) \right), \label{eq:comb1}
\end{equation}
 where $g_c: [0,1]^M\rightarrow \mathcal{R}$ is a continuous function that is non-decreasing in each coordinate, $G_c: \mathcal{R}\rightarrow [0,1]$ is the cumulative distribution function (CDF) of $g_c(U_1,\cdots,U_M)$ where $U_1,\cdots, U_M$ are i.i.d. Uniform[0,1] random variables. 

Similarly, the $M$ $p$-value functions $p_1^{U+}(\theta),  \ldots, p_M^{U+}(\theta)$ defined by  (\ref{eq:p_right_U}) can be combined into a single function 
\begin{equation}
p_c^{U+}(\theta) = G_c \left( g_c \left (p_1^{U+}(\theta),  \ldots, p_M^{U+}(\theta) \right) \right). \label{eq:comb2}
\end{equation}

\begin{proposition} \label{prop:combined}
The combined $p$-value functions $p_c^{L+}(\theta)$ and $p_c^{U+}(\theta)$ are respectively lower and upper CDs as per Definition \ref{def:upperlowerCD}.
\end{proposition}

As a consequence of Proposition \ref{prop:combined}, the combined $p$-value functions $p_c^{L+}(\theta)$ and $p_c^{L-}(\theta) = 1 - p_c^{U+}(\theta)$ can be used to test the sharp null $H_0^{\theta}$ for any specific value of $\theta$ against one-sided alternatives. Further, similar to (\ref{eq:p_twosided}), the combined two-sided $p$-value function 
\begin{equation}
p_c^{L}(\theta) = 2 \min \left\{ p_c^{L+}(\theta) , p_c^{L-}(\theta) \right\} \label{eq:comb3}
\end{equation}		
can be used to draw inference about $\theta$ by testing the sharp null hypothesis against a two-sided alternative. Finally, for $0<\alpha<1$, define $\theta_{\ell,c} = \sup_\theta\{\theta: p_c^{L+}(\Dobs,\theta) \le \alpha/2 \}$ and $\theta_{u,c} = \inf_\theta\{\theta: p_c^{L-}(\Dobs,\theta) \le \alpha/2\}$. Then by part (c) of Proposition \ref{prop:CI}, the interval $[\theta_{\ell,c}, \theta_{u,c})$ is a $100(1-\alpha)\%$ interval for $\theta$ obtained by combining the $M$ studies.

To implement the above steps, we need to choose specific forms of the function $g_c(\cdot)$. \cite{xie2011confidence} showed that the form $g_c(u_1, \ldots, u_M) = \sum_{i=1}^M w_i F_0^{-1}(u_i)$, where $F_0(\cdot)$ is a CDF and $w_1, \ldots, w_M$ are non-negative weights with at least one $w_i \ne 0$, generates most classical methods for combining $p$-values. Three examples are given below. 
\begin{enumerate}
\item With $w_i=1$ for all $i=1, \ldots, M$ and $F_0(x) = \Phi(x)$, the CDF of the standard normal distribution one obtaines Stouffer's method, in which
\begin{equation}
p_c^{L+}(\theta) = \Phi \left[ \frac{1}{\sqrt{M}} \sum_{i=1}^M \Phi^{-1} \left(p_i^{L+}(\theta) \right) \right], \ p_c^{U+}(\theta) = \Phi \left[ \frac{1}{\sqrt{M}} \sum_{i=1}^M \Phi^{-1} \left(p_i^{U+}(\theta) \right) \right]. \label{eq:stouffer_comb} 
\end{equation}	
\item Similarly, with $w_i=1$ for all $i=1, \ldots, M$ and $F_0(x) = e^x$ for $x \le 0$ generates Fisher's method, in which
\begin{equation}
p_c^{L+}(\theta) = P \left[ \chi^2_{2M} \ge -2 \sum_{i=1}^M \log\left(p_i^{L+}(\theta) \right) \right], \ p_c^{U+}(\theta) = P \left[ \chi^2_{2M} \ge -2 \sum_{i=1}^k \log\left(p_i^{U+}(\theta) \right) \right]. \label{eq:Fisher_comb}
\end{equation}
\item Finally, again taking $w_i=1$ for all $i=1, \ldots, M$ and $F_0(x) = \frac{1}{2}e^t\mathbbm{1}_{(x\leq 0)}+(1-\frac{1}{2}e^{-x})\mathbbm{1}_{(x>0)}$, i.e., the double exponential or Laplace CDF instead of the negative exponential CDF leads to the double exponential (DE) method for combining $p$-values.
\end{enumerate}

\cite{Singh2005} showed that the DE method for combining $p$-values is Bahadur efficient on both tails, whereas the Fisher method is Bahadur efficient only on the lower-sided tail. The proposed approach for combining FRT-based inference from independent randomized experiments will be demonstrated in Sections \ref{sec:simulations} and \ref{sec:example} using Fisher's and double exponential methods described above. It may be noted though, that the principles and algorithms described in this section opens up a multitude of possibilities for meta analysis.  

\section{Monte Carlo approximation of the $p$-value function} \label{sec:MonteCarlo}

The $p$ value functions defined in Section \ref{ss:one-sided} and Section \ref{ss:two-sided} can be computed for any given value of $\theta$ if all possible realizations $\Wrep$ of the assignment vector $\mathbf{W}$ can be obtained leading to generation of the exact randomization distribution of the test statistic $T(\Yimptheta, \Wrep)$. However, even for a moderate population size the total number of possible realizations of $\mathbf{W}$ is typically computationally prohibitive. The common solution to this problem is to draw, repeatedly and independently, randomized treatment assignment vectors $\Wrep_1, \ldots, \Wrep_M$, and obtain a Monte Carlo estimate of the $p$-value function based on the values of the test statistic computed from these $M$ draws. Consider specifically the estimation of $p^{L+}(\Dobs, \theta)$ defined in (\ref{eq:p_right_L}). The Monte-Carlo estimator of $p^{L+}(\Dobs, \theta)$ based on $K$ Monte Carlo samples is given by 
\begin{equation}
\widehat{p}_K^{L+}(\Dobs, \theta) = \frac{1}{K} \sum_{k=1}^M \mathbb{I} \left( T(\Yimptheta, \Wrep_k) \ge T(\Dobs) \right), \label{eq:p_est}
\end{equation}
where $\mathbb{I}(A)$ is the indicator function for event $A$. All the other $p$-value functions can be estimated similarly. In spite of this estimator being used since the times of Fisher, the effect of the Monte Carlo sample size $K$ on the accuracy of the estimator $\widehat{p}^{L+}(\Dobs, \theta)$ has not been researched. Therefore, it is not quite clear what should be a ``reasonable'' $K$ for a completely randomized design with $N = 20$ or a matched pair design with $N = 30$. Below, we provide a new result in the form of a concentration inequality that sheds some light on this question. 
\begin{theorem} \label{thm:upper_bound}
Let $K$ denote the number of Mone Carlo samples from the distribution of $\mathbf{W}$ and let $\widehat{p}^{L+}(\Dobs, \theta)$ be as defined in (\ref{eq:p_est}), where the underlying test statistic $T = T(\Yimptheta, \Wrep)$ is a non-decreasing and right continuous function of $\theta$ for fixed $\Wrep$. Fix $\epsilon > 0$. Then,
\begin{equation}
P \left(\sup_\theta \left|\widehat{p}_K^{L+}(\Dobs, \theta)- p^{L+}({\Dobs},\theta) \right| > \epsilon \right)\leq \min \left\{ 1, 4e^{-\frac{K\epsilon^2}{8}} \right\}. \label{eq:upper_bound}
\end{equation}
\end{theorem}

It is important to note that the bound (\ref{eq:upper_bound}) does not depend on $N$, making it particularly useful for cases when the total number of possible assignments is large. Table \ref{tab:thresholdk} presents threshold values of $K$ required to attain a probability bound of 0.01 for different values of estimation error $\epsilon$. In other words, when $K \ge K_{\text{thres}}$, the probability bound in the right hand side of \ref{eq:upper_bound} is smaller than or equal to 0.01.

\begin{table}[ht]
			\centering \small
			\caption{Threshold $K$ required to attain a probability bound of 0.01 for different values of $\epsilon$} \label{tab:thresholdk}
			\begin{tabular}{l|c}     
			\hline
				$\epsilon$ & $K_{\text{thres}}$ \\ \hline
        0.1  & 4794  \\
        0.05 & 19173 \\
        0.02 & 119830\\
        0.01 & 479318 \\
      0.005  & 1917269 \\
      0.002  & 11982930 \\
      0.001  & 47931717 \\ \hline
			\end{tabular}
		\end{table}

The threshold table \ref{tab:thresholdk} provides us with a procedure to determine the Monte Carlo sample size $K$ that guarantees that the bound derived in Theorem \ref{thm:upper_bound} does not exceed 0.01 for a specified level of accuracy $\epsilon$, defined in terms of the supremum of the absolute difference between the $p$-value function and its Monte Carlo estimate. If the total number of all possible assignments, say $K^*$ is less than or equal to $K$, then one can enumerate all of them to obtain the $p$-value function. However, if $K^* > K_{\text{thres}}$, then a simple random sample of $K_{\text{thres}}$ assignments with replacement will be a reasonable strategy. 

\section{Simulations} \label{sec:simulations}

In this section we conduct simulations to demonstrate that our proposed guidelines and algorithms for estimating the $p$-value functions, inverting them to obtain confidence intervals and combining inferences across multiple independent experiments produce the desired results. We consider two types of randomized experiments: the completely randomized design (CRD) and the randomized block design (RBD). In the former, an even number $N$ of experimental units are equally split into treatment (denoted by 1) and control (denoted by 0) groups at random. In the latter, we consider $b$ blocks of experimental units with an equal even number ($k$) of units in each block (block size), so that $N = bk$ is the total number of units. The $k$ units within each block are equally split into treatment and control groups at random. Note that $b=1$ for an RBD is equivalent to a CRD.

We consider several scenarios shown in Table \ref{tab:coverage-width}, in each of which we consider combining results from two experiments with design parameters $(b_1, k_1)$ and $(b_2, k_2)$, where for $j = 1,2$, $b_j$ and $k_j$ denote the number of blocks and the block size respectively. The two individual experiments are either CRD or RBD.

For each individual experiment across all scenarios, the potential outcomes under control, $Y_i(0)$, $i=1, \ldots, N$ are generated from a lognormal distribution with parameters 0 and 1. The true additive effect is assumed to be zero, so that $Y_i(1) = Y_i(0)$ for $i=1, \ldots, N$. Potential outcomes once generated are kept fixed. The units are assigned to treatments in a manner described earlier, depending on whether the design is CRD or RBD.

Next, for each experiment, FRT is conducted using the difference of averages statistics between treatment and control groups as the test statistic. Denoting by $M$ all possible assignments, $K=\min\{M, 10000\}$ permutations are used to calculate or estimate the $p$-value functions $p^{L+}(\bf D^{\text{obs}},\theta)$ and $p^{U+}(\bf D^{\text{obs}},\theta)$. For each individual experiment, 95\% confidence intervals are obtained using the method described in part (c) of Proposition \ref{prop:CI} with $\alpha_1 = \alpha_2 = .025$. Finally the $p$-value functions from the two experiments in each scenario are combined using Fisher's method given by (\ref{eq:Fisher_comb}) and the double exponential (DE) method, and the 95\% confidence intervals are generated using the combined two-sided $p$-value function $p_c^{L}(\theta)$ defined in (\ref{eq:comb3}).

The simulation for each scenario is repeated 1500 times to calculate the coverage and the distribution of width of the 95\% intervals generated from the individual and combined experiment and results are shown in Table \ref{tab:coverage-width}.

\begin{table}[!ht]
	\centering
	\caption{Coverage and width of a 95\% CI by Fisher's and DE $p$-value combinations over repetitions} \label{tab:coverage-width} \vspace{0.1cm}
	\begin{threeparttable}[t]
		\centering \scriptsize
		\begin{tabular}{ccccccccccccccc}   
			\hline    
			\multicolumn{5}{c}{Scenario} & & \multicolumn{4}{c}{Coverage} & &\multicolumn{4}{c}{length of CI (mean$\pm$sd)} \\ 
			\cline{1-5}\cline{7-10}\cline{12-15}
			{Designs 1 \& 2}& $b_1$ & $k_1$ & $b_2$ & $k_2$ &  & Exp 1 & Exp 2 & Fisher & DE &  & Exp 1 & Exp 2 & Fisher  & DE \\ 
			\hline
			\multirow{10}{*}{CRD  \& CRD}& 1 &10 & 1 & 10 &\multirow{18}{*}{\small } & 0.952 & 0.952 &0.955 &0.958  &\multirow{18}{*}{\small } & 6.01$\pm$0.44 & 6.01$\pm$0.44 &4.20$\pm$0.31 & 4.28$\pm$0.36\\
			&1 &16 &1 & 16 & & 0.950  & 0.950  &0.951 &0.947 & & 3.70$\pm$0.16 & 3.70$\pm$0.16 &2.69$\pm$0.15 &  2.65$\pm$0.27\\ 
			&1 &24 &1 & 24 & & 0.943  & 0.943 &0.953 &0.945 &  & 2.53$\pm$0.07& 2.53$\pm$0.07  &1.88$\pm$0.09 &  1.81$\pm$0.20\\
			&1 &30 &1 & 30 & & 0.951 & 0.951 &0.955 &0.955 & & 2.25$\pm$0.05 & 2.25$\pm$0.05 &1.66$\pm$0.08 &1.63$\pm$0.17
			\\ \cline{2-5}\cline{7-10}\cline{12-15}
			&1 &10 &1 & 16 & & 0.952 & 0.950 & 0.952 &0.957 & & 6.01$\pm$0.44 & 3.70$\pm$0.16 & 3.32$\pm$0.39&3.23$\pm$0.30 \\ 
			&1 &10 &1 & 24 & & 0.952 & 0.943 & 0.948  &0.957 &  & 6.01$\pm$0.44 & 2.53$\pm$0.07 & \textbf{2.73$\pm$0.55}&2.45$\pm$0.21\\ 
			&1 &10 &1 & 30 & & 0.952 & 0.951 &0.942  &0.943 & & 6.01$\pm$0.44 & 2.25$\pm$0.05 & \textbf{2.57$\pm$0.60}&2.23$\pm$0.18\\ 
			&1 &16 &1 & 24 & & 0.950 & 0.943 &0.946  &0.950 & & 3.70$\pm$0.16 & 2.53$\pm$0.07 &2.24$\pm$0.24&2.14$\pm$0.23\\
			&1 &16 &1 & 30 & & 0.950 & 0.951 &0.955  &0.954 &  &3.70$\pm$0.16 & 2.25$\pm$0.05 &2.10$\pm$0.27&1.98$\pm$0.20\\
			&1 &24 &1 & 30 & & 0.943 & 0.951 &0.950  &0.949 & & 2.53$\pm$0.07& 2.25$\pm$0.05 &1.77$\pm$0.11&1.7$\pm$0.21\\\cline{1-5}\cline{7-15}
			\multirow{6}{*}{RBD \ \& RBD } &2 & 8 &2 & 8 & &0.951 & 0.951 &0.947 &0.948 & &3.78$\pm$0.19&3.78$\pm$0.19 &2.72$\pm$0.16&2.72$\pm$0.26\\
			&4 & 4 &4 & 4 & &0.951 &0.951 &0.962 &0.957 &  &3.23$\pm$0.16 &3.23$\pm$0.16 &2.32$\pm$0.13&2.30$\pm$0.26 \\ 
			&8 & 2 &8 & 2 & &0.953 &0.953 &0.961 &0.963 & &2.81$\pm$0.25 &2.81$\pm$0.25&1.96$\pm$0.15&2.01$\pm$0.16 \\ \cline{2-5}\cline{7-10}\cline{12-15}
			&2 & 8 &4 & 4 & &0.951 & 0.951 &0.951&0.955  & &3.78$\pm$0.19 &3.23$\pm$0.16 &2.51$\pm$0.18&2.49$\pm$0.26 \\
			&2 & 8 &8 & 2 & &0.951 & 0.953 &0.955 &0.955 & &3.78$\pm$0.19 &2.81$\pm$0.25 &2.33$\pm$0.23&2.33$\pm$0.19 \\
			&4 & 4 &8 & 2 & &0.951 &0.953 &0.955&0.955 &  &3.23$\pm$0.16 &2.81$\pm$0.25 &2.15$\pm$0.17&2.16$\pm$0.20 \\ \cline{1-5}\cline{7-15}
			\multirow{9}{*}{ CRD  \& RBD } 
			&1 & 10 &2 & 8 & &0.952 & 0.951 &0.944 &0.943 &  &6.01$\pm$0.44 & 3.78$\pm$0.19 &3.34$\pm$0.374 &3.28$\pm$0.29\\ 
			&1 & 10 &4 & 4 & &0.952 &0.951 &0.945 &0.953 & &6.01$\pm$0.44 & 3.23$\pm$0.16 &3.09$\pm$0.453&2.93$\pm$0.29\\ 
			&1 & 10 &8 & 2 & &0.952 &0.953 &0.945 &0.953 & &6.01$\pm$0.44 & 2.81$\pm$0.25  & \textbf{2.92$\pm$0.49}&2.77$\pm$0.23\\ 
			&1 & 16 &2 & 8 & &0.950 & 0.951 &0.942 &0.946 & &3.70$\pm$0.16 & 3.78$\pm$0.19 &2.70$\pm$0.16&2.69$\pm$0.26\\ 
			&1 & 16 &4 & 4 & &0.950 &0.951 &0.955 &0.953 & &3.70$\pm$0.16 & 3.23$\pm$0.16 &2.50$\pm$0.16&2.46$\pm$0.26 \\ 
			&1 & 16 &8 & 2 & &0.950 &0.953 &0.957 &0.957 & &3.70$\pm$0.16 & 2.81$\pm$0.25  &2.32$\pm$0.22&2.31$\pm$0.19\\ 
			&1 & 30 &2 & 8 & &0.951 & 0.951 &0.943 &0.945 & &2.25$\pm$0.05 & 3.78$\pm$0.19 &2.10$\pm$0.28&1.99$\pm$0.20\\ 
			&1 & 30 &4 & 4 & &0.951 &0.951 &0.951 &0.953 & &2.25$\pm$0.05 & 3.23$\pm$0.16 &1.96$\pm$0.19&1.88$\pm$0.20 \\ 
			&1 & 30 &8 & 2 & &0.951 &0.953 &0.950 &0.948  & &2.25$\pm$0.05 & 2.81$\pm$0.25  &1.80$\pm$0.15&1.79$\pm$0.16\\  
			\hline
		\end{tabular}
		\end{threeparttable}
  \end{table}

The simulations provide empirical evidence of the theoretical result that the proposed method for inverting FRT to obtain confidence intervals produces intervals with the desired coverage for individual as well as combined experiments. There are some instances where the coverage falls slightly below 95\% and these cases can be attributed to the error involved in estimation of the $p$-value functions with 10000 Monte Carlo samples. 

It is natural to expect that the length of the interval generated by combining the two individual experiments would be shorter than the lengths of interval obtained from \emph{each} individual experiment as such a fusion should increase the precision of inference. From Table \ref{tab:coverage-width} we see that this expectation is fulfilled in most scenarios except three (shown in bold), where the mean width of the interval obtained from the combined experiment is marginally greater than the mean width of the interval obtained from the larger experiment. These counter-intuitive results appear to hold in situations where the two experiments differ substantially in sample size and Fisher's method is used for combining experiments. A plausible explanation of this anomaly is the observation by \cite{Singh2005} that Fisher's method of combining $p$-values sharpens the inference for only tail of CDs, whereas the DE method sharpens the inference on both tails. These results suggest that the double-exponential method may be more robust than Fisher's method, particularly when it comes to combining experiments of different sizes.

\section{Real data example} \label{sec:example}

We consider an adopted example from \cite{Shadish:2008} involving estimating the effect of a vocabulary training course on vocabulary test scores to demonstrate the proposed methodology. The overall goal of the study, that includes a randomized experiment and an observational study, was to examine whether observational studies can be analyzed to yield valid estimates of a causal effect. Here, for demonstrating our approach, we only consider the randomized experiment and the problem of estimating the causal effect that taking a vocabulary training course has on vocabulary test scores. We split the 235 experimental units into two groups consisting of 100 and 135 units, and pretend that two independent completely randomized experiments were conducted, in each of which the treatment (vocabulary training) or control (math training) were assigned to the units using a completely randomized treatment assignment. In the first experiment, 44 out of 100 units were assigned to treatment, whereas in the second experiment, 72 out of 135 units were assigned to treatment. The outcome was the vocabulary test score after the experiment. The 95\% confidence intervals for the average treatment effect obtained from the two individual experiments and the combined experiment are shown in Table \ref{tab:shadish}. Three methods - Fisher's, Stouffer's and double exponential - are used to combine the results, and are found to produce almost identical results.

\begin{table}[ht] 
\caption{Combining information from two randomized experiments} \label{tab:shadish}
\centering \small
\begin{tabular}{cccccccl} \hline
Experiment & $N$ & $N_0$ & $N_1$ & $\ybarone$ & $\ybarzero$ & $\hat{\tau}$ & \hspace{.2 in} 95\% CI \\ \hline
1          & 100 & 44    & 56    &   16.55 &  8.36 & 8.19 &  (6.893, 9.476) \\
2          & 135 & 72    & 63    &   15.97 &  7.83 & 8.14 &  (6.939, 9.351) \\ \hline
           &     &       &       &         &       &      &  (7.190, 8.923) \ (Fisher's method) \\
Combined	 & 235 & 116   & 119   &   16.19 & 8.08  & 8.11 &	 (7.290, 9.040) \ (Stouffer's method) \\
	         &     &       &       &         &       &      &  (7.200, 9.125) \ (Double exponential method) \\
\hline
\end{tabular}
\end{table}

\section{Discussion and future directions} \label{sec:discussion}

In most scientific studies, assessing causal relationships among variables is considered more important than studying associations. The distinct difference between association and causality is now well understood: causality can only be determined by utilizing known or assumed knowledge about how the data were collected and consequently is more difficult to establish than association. However, technology has now created a perfect platform to design and analyze large studies conducted to assess causal effects of interventions and the FRT, with its unique ability to facilitate model-free assessment of causal effects is expected to have tremendous potential in modern day experiments. In this article, we attempt to address some apparently ambiguous aspects associated with the methodology for computation, inversion and principled aggregation of FRTs by developing a unified and comprehensive framework based on the versatile inferential tool confidence distribution (CD).

One of the main criticisms of FRT has been the sharp null hypothesis, that many, including Neyman had considered overly strong leading to the infamous Neyman-Fisher debate in 1935 \citep{Sabbaghi2014}. However, the possibility of applying FRT to test assess weaker null hypothesis has been explored and identified by a few researchers - see for example, \cite{DingDasgupta:2018}, Ding and Dasgupta (2018), \cite{Caughey:2017}, and \cite{WuDing:2019}. In particular, \cite{Caughey:2017} showed that the interval estimators obtained by inverting FRT can be interpreted more meaningfully under a bounded null hypotheses if EI test statistics are used. It will be interesting to extend our results to such weaker hypotheses and consequently have broader interpretations of the interval estimators.

We believe this article will open up a number of research possibilities. First, all our results pertain to finite samples. Exploring asymptotic properties of the interval estimators for individual and combined experiments using finite population asymptotics \citep{LiDing:2017} and borrowing relevant literature from CD literature will be a useful direction. Second, exploring ways to optimally combine experiments, as discussed in the last paragraph of Section \ref{sec:meta_analysis} will be an interesting line of investigation. Third, as seen from Table \ref{tab:thresholdk}, the Mone Carlo sample size required to achieve a small maximum absolute estimation error of say, .001 with a probability of .01 is approximately $4.79 \times 10^7$. Such a computation may still take a prohibitive amount of time, specifically (a) if its needs to be conducted repeatedly in a simulation setting and (b) if the sampling is of acceptance-rejection type as in rerandomization \cite{Morgan:2012} type of settings. A popular approach to tackle such computational issues is to use a split-and-conquer approach. The CD is well known for providing statistical guarantees to split-and-conquer approaches to inference problems \citep{Chen2014}, paving the way for exploring split-and-conquer approaches to increase computational efficiency in FRT. Fourth, extending the FRT-CD framework for analysis of data from observational studies and conducting sensitivity analysis is an interesting possibility. Finally, combining experiments and observational studies is an area of growing interest, and the FRT-CD may provide an excellent foundation for this area of research.

\section*{Acknowledgements}
This research was partially supported by National Science Foundation Grant Number DMS 1451817.

	\bibliographystyle{apalike}
	\bibliography{CausalCDreferences,Unified}
	
\newpage	
	
\setcounter{section}{1}	
\section*{Supplementary materials}

\subsection{Demonstrating steps of FRT and its inversion} \label{sm:toy}

\begin{example}[A toy example] \label{example:toy1}
		
		The second and third columns of Table \ref{tab:toy} constitute the $10 \times 2$ potential outcomes matrix $\Ytrue$ in which the potential outcomes for the control are generated independently from a lognormal distribution with parameters 0.5 and 1, and $Y_i(1) = Y_i(0) + 1$ for $i=1, \ldots, 10$, making the true value of $\theta$ equal to 1.
		
		\begin{table}[ht]
			\centering
			\caption{Potential outcomes and observed data in a toy example} \label{tab:toy}
			\begin{tabular}{c|cc|c|c}       
				Unit($i$) & $Y_i(0)$ & $Y_i(1)$ & $W_i^{\textrm{obs}}$ & $Y_i^{\textrm{obs}}$ \\ \hline
				1  & 1.00 & 2.00 & 1 & 2.00 \\
        2  & 1.88 & 2.88 & 1 & 2.88 \\
        3  & 1.52 & 2.52 & 1 & 2.52 \\
        4  & 4.00 & 5.00 & 1 & 5.00 \\
        5  & 1.85 & 2.85 & 0 & 1.85 \\
        6  & 2.27 & 3.27 & 0 & 2.27 \\
        7  & 0.92 & 1.92 & 0 & 0.92 \\
        8  & 3.37 & 4.37 & 0 & 3.37 \\
        9  & 0.72 & 1.72 & 1 & 1.72 \\
       10  & 1.15 & 2.15 & 0 & 1.15 \\ \hline
			\end{tabular}
		\end{table}
		
		\begin{figure}[ht]
			\centering \caption{$p$-values for four different sharp null hypotheses} \label{fig:Trep_hist}
			\includegraphics[width=12cm]{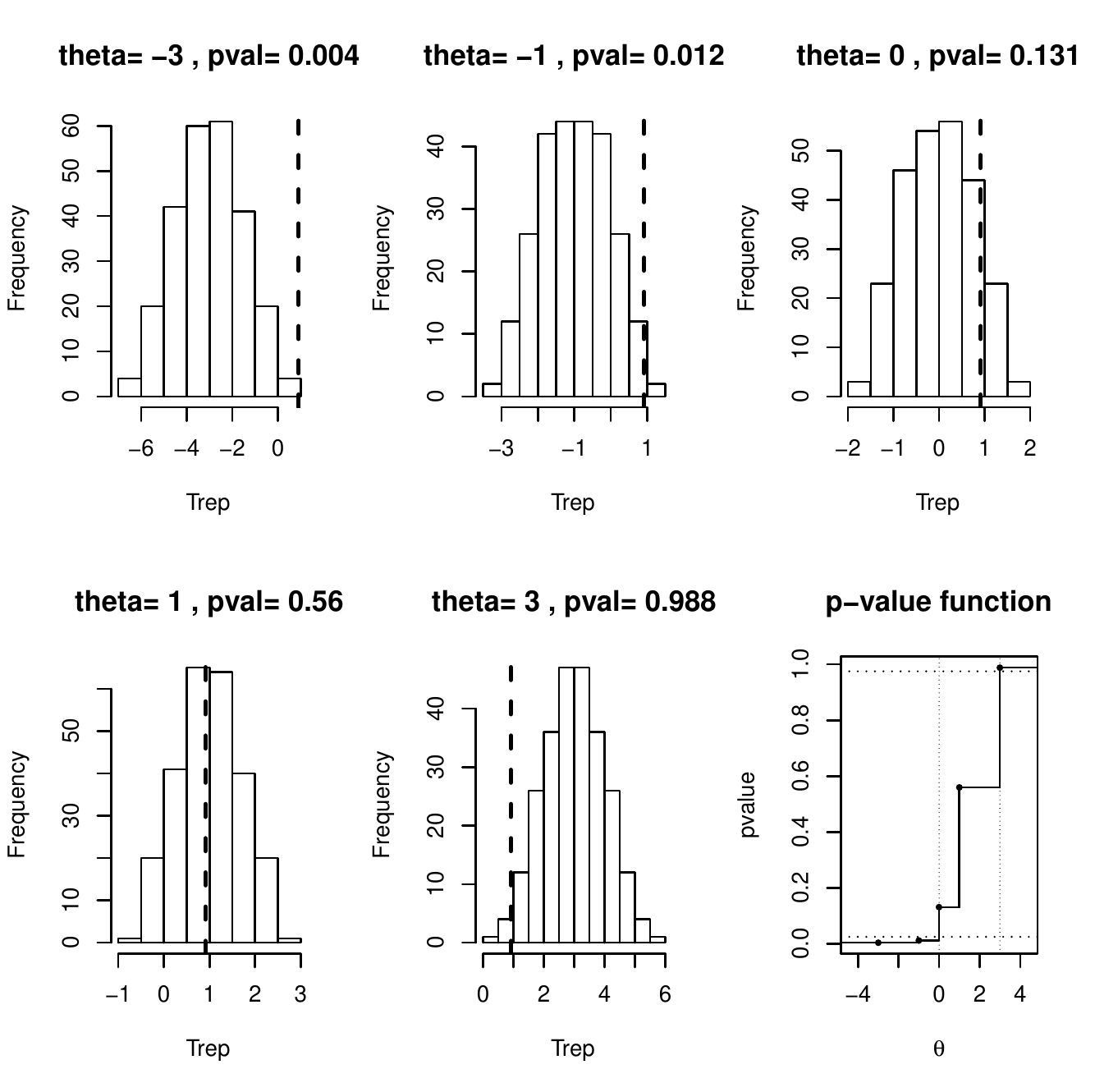}
		\end{figure}
		
Assume the assignment mechanism to be a balanced completely randomized design that assigns five units to control and the rest to treatment. The observed assignment $\Wobs = (1,1,1,1,0,0,0,0,1,0)$ shown in the fourth column of Table \ref{tab:toy} generates the observed data $\Yobs$ in the fifth column. Suppose we are interested in testing the sharp null hypothesis $H_0^{\theta}$ for five values of $\theta$: -3, -1, 0, 1 and 3. Note that the case with $\theta=0$ is the sharp null hypothesis of no treatment effect, i.e., $H_0: \theta = 0$. Consider the test statistic $T = \ybarone - \ybarzero$, the observed value of which is 0.912. The distributions of $\Trep = T(\mathbf{Y}^{\textrm{imp}}, \Wrep)$ over the 252 possible draws $\Wrep$ (all possible permutations of the $\Wobs$) for the five hypotheses are shown in Figure \ref{fig:Trep_hist}, yielding observed $p$-values of 0.004, 0.012, 0.131, 0.56 and 0.988 respectively. Thus one will reject $H_0^{\theta}$ for $\theta = -3, -1$ at 5\% level of significance and not reject it for $\theta = 0,1, 3$. Figure \ref{fig:Trep_hist} also suggests that (i) the $p$-values would have remained unchanged if one had used a centered test statistic $T = \ybarone - \ybarzero - \theta$ instead, and (ii) for the given data, it is possible to obtain a ``$p$-value function'' of $\theta$ shown in the lower right panel of the figure, by testing a set of sharp null hypotheses. This step function can be inverted to generate a 95\% interval $[0,3]$ (as shown by dotted lines) for the true additive effect $\theta$. As we shall see in Section \ref{sec:connection}, most of the subsequent developments will be based on this function and its variants.
		
	\end{example}

\subsection{Illustration of a CD function and CV} \label{sm:CD}

	\begin{figure}[ht]
		\centering
		\caption{\small (a) A CD function $H_n(\theta)$ and (b) its corresponding {\it confidence curve} (CV) $CV_n(\theta)= 2 \min\{H_n(\theta), 1-H_n(\theta)\}$ for parameter $\theta$.} \label{fig:CD-CV}
		\centering{\includegraphics[width=15 cm]{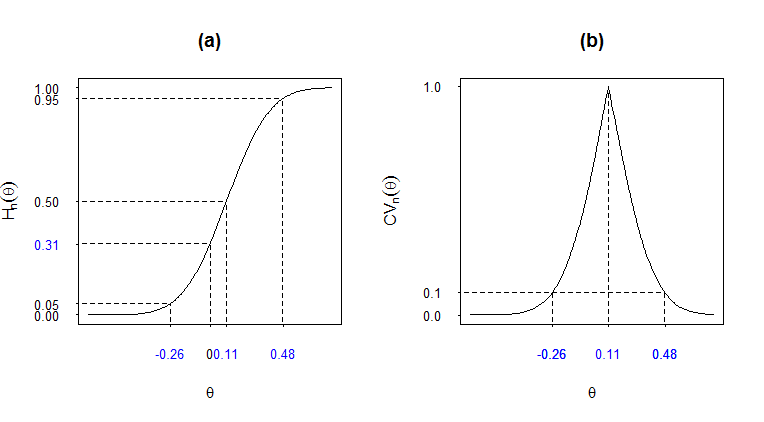}}
		\singlespace
	\end{figure} 

Figure \ref{fig:CD-CV}~(a) illustrates a CD function $H_n(\theta) = \Phi(\sqrt{n}(\theta-\bar{y}))$, based on a sample from an $N(\theta,1)$  distribution with size $n=20$ and sample mean $\bar y = 0.11$ and demonstrates how it can be used to draw inference about the normal mean $\theta$. The dashed lines illustrate how to obtain a point estimate of $0.11$, a 90\% confidence interval of $(-0.26, 0.48)$, and a $p$-value $0.31$ for testing the hypothesis $H_0:\theta \le 0$ versus $H_1:\theta>0$. Figure \ref{fig:CD-CV}~(b) shows the corresponding {\it confidence curve} (CV) $CV_n(\theta)= 2 \min\{H_n(\theta), 1-H_n(\theta)\}$ for parameter $\theta$ and demonstrates how similar inferential information can be obtained from it.

\subsection{Proof of results} \label{sm:proofs}

\subsubsection*{Proof of Proposition \ref{prop:FRTlowerCD}} \label{pf:Prop1}
\begin{proof}
First, by (\ref{eq:p_right_L}),
			\begin{eqnarray} 
				p^{L+}(\Dobs,\theta_0) &=& P \left( T(\Yimpthetazero, \Wrep) \ge T(\Dobs)  \right) \nonumber \\
				&=& P \left( T(\Yimpthetazero, \Wrep) \ge T(\Ytrue, \Wobs)  \right) \ \mbox{by} \ (\ref{eq:Tobs}) \nonumber \\ 
				&=& P \left( T(\Ytrue, \Wrep) \ge T(\Ytrue, \Wobs)  \right) \ \mbox{since} \ \Yimpthetazero=\Ytrue. \label{eq:prop1proof1} 
			\end{eqnarray}
If $T_{(1)} < T_{(2)} < \cdots < T_{(m)}$ denote the $m$ unique ordered values of $T$ for $\theta = \theta_0$ and $\gamma_i=P \left(T(\Ytrue, \bm W)=T_{(i)}\right)>0$ for $i=1,2,\cdots, m$, so that $\gamma_1+\gamma_2+\cdots+\gamma_m=1$.  Then by (\ref{eq:prop1proof1}), $p^{L+}(\Dobs,\theta_0)$ can take values $\gamma_m, \gamma_m+\gamma_{m-1}, \cdots, \gamma_m+\gamma_{m-1}+\cdots+\gamma_{1}$ with probabilities $\gamma_m,\gamma_{m-1},\cdots \gamma_1$ respectively.
			
For any $\alpha \in (0,1)$, there exists a unique $0 \le j\le m-1$ such that $\alpha \in (\gamma_m+\cdots +\gamma_{m-j+1}, \gamma_m+\cdots+\gamma_{m-j})$, where $\gamma_{m+1} = - \gamma_{m}$. Consequently, 
			\begin{equation}
			P \left( p^{L+}(\Dobs,\theta_0) \le \alpha \right) = \sum_{k=0}^j P \left( p^{+}(\Dobs,\theta_0)=\gamma_m+\cdots+\gamma_{m-k+1} \right) =\gamma_m+\cdots+\gamma_{m-j+1} \le \alpha. \label{eq: domination}
			\end{equation}
This establishes that $p^{L+}(\Dobs,\theta_0)$ stochastically dominates the Uniform[0,1] random variable, proving part (a) of the Proposition for $p^{L+}(\Dobs,\theta_0)$. Further, by (\ref{eq: domination}), it follows that for a fixed $\alpha \in (0,1)$,
$$\alpha - P \left( p^{L+}(\Dobs,\theta_0) \le \alpha \right) = \alpha - (\gamma_m + \ldots + \gamma_{m-j+1}) \le \gamma_{m-j}$$ for some $0 \le j \le m-1$. Consequently, for any arbitrary $\alpha \in (0,1)$, 
$$\alpha - P \left( p^{L+}(\Dobs,\theta_0) \le \alpha \right) \le \max \{\gamma_1, \ldots, \gamma_m\}.$$
This proves part (c) for $p^{L+}(\Dobs,\theta_0)$. 

Similar arguments lead to proof of (a) for $p^{L-}(\Dobs,\theta_0)$, (b) for $p^{U+}(\Dobs,\theta_0)$ and $p^{U-}(\Dobs,\theta_0)$ and part (c) for $p^{L-}(\Dobs,\theta_0)$, $p^{U+}(\Dobs,\theta_0)$ and $p^{U-}(\Dobs,\theta_0)$
		\end{proof}

\subsubsection*{Proof of Theorem 1} \label{pf:Thm1}

\begin{proof}
\noindent Part (a): Let $\Yimpthetai(1)$ and $\Yimpthetai(0)$ denote the imputed potential outcomes for unit $i$ under $H_0^{\theta}$. Then 
\begin{eqnarray}
			\Yimpthetai(1)=\left\{ \begin{array}{cc} \Yobs_i(1), & \Wobs_i=1 \\
			\Yobs_i(0)+\theta, & \mbox{otherwise} \end{array} \right. \ \ \text{and} \ \ \Yimpthetai(0)=\left\{ \begin{array}{cc} \Yobs_i(1)-\theta, & \Wobs_i=1 \\
			\Yobs_i(0), & \mbox{otherwise} \end{array} \right. \label{eq: imputed potential outcome}
\end{eqnarray}
Thus, for any given assignment $\bm W$, the $N \times 1$ vectors of imputed potential outcomes $\Yimptheta(1) $ and $\Yimptheta(0)$ are respectively non-decreasing and non-increasing in $\theta$ in the sense of Definition 5. Since $T$ is EI, by Definition 6, $T(\Yimptheta, \Wrep) = T(\Yimptheta(1), \Yimptheta(0), \Wrep)$ is non-decreasing in $\Yimptheta(1)$ and non-increasing in $\Yimptheta(0)$, and consequently non-decreasing in $\theta$. Then by  (\ref{eq:p_right_L}), $p^{L+}(\Dobs,\theta_0)$ is non decreasing in $\theta$. 
			
			\medskip

\noindent Part (b):	First, note that, for any $\theta$, $T(\Yimptheta, \Wobs)=T(\Ytrue,\Wobs)=T(\Dobs)$. Moreover, as $\theta\rightarrow \infty$, for any fixed assignment $\bm W \neq \Wobs$, by the stated condition it follows that $T(\Yimptheta,\bm W) \rightarrow \infty> T(\Dobs)$. Thus, as $\theta\rightarrow \infty$, $T(\Yimptheta,\bm W)\ge T(\Dobs)$ for all $\bm W$ and therefore by  (\ref{eq:p_right_L}),  $p^{L+}(\Dobs, \theta)\rightarrow 1$. Similarly, by the stated condition, it follows that for any fixed assignment $\bm W \neq \Wobs$ , as $\theta\rightarrow -\infty$, $T(\Yimptheta,\bm W) \rightarrow -\infty < T(\Dobs)$. However, since $T(\Yimptheta, \Wobs)=T(\Dobs)$, by (\ref{eq:p_right_L}) we have that, $p^{L+}(\Dobs, \theta) \rightarrow P(\bm W=\Wobs)$ as $\theta\rightarrow -\infty$. This proves (i).
			
			\medskip
			
\noindent Finally, by definition, $p^{+}(\Dobs,\theta)$ is a step function with jump discontinuities, and right continuous for fixed $\Dobs$ since $T$ is right continuous. This proves (ii)
			
			\medskip

\end{proof}

\subsubsection*{Proof of Proposition \ref{prop:CI}} \label{ss:prop2proof}

\begin{proof}
Here we prove part (c), and the proofs of parts (a) and (b) are similar. By definition of $\theta_{ell}(\alpha_1)$ and $\theta_u(\alpha_2)$, it follows that the coverage probability
\begin{eqnarray*}
&& P_{\theta} \left[ \theta_{\ell}(\alpha_1) \le \theta < \theta_u(\alpha_2) \right] \\
&\ge& P_{\theta} \left[p^{L+}(\Dobs, \theta) > \alpha_1 \ \text{and} \ p^{L-}(\Dobs, \theta) > \alpha_2\right] \\
&\ge& P_{\theta} \left[p^{L+}(\Dobs, \theta) > \alpha_1\right] + P_{\theta} \left[p^{L-}(\Dobs, \theta) > \alpha_2\right] - 1 \\
&=& 1 - P_{\theta} \left[p^{L+}(\Dobs, \theta) \le \alpha_1\right] + 1 - P_{\theta} \left[p^{L-}(\Dobs, \theta) \le \alpha_2\right] - 1 \\
&\ge& 1 - \alpha_1 + 1 - \alpha_2 -1 \\
&=& 1 - \alpha.
\end{eqnarray*}

\end{proof}

\subsubsection*{Proof of Proposition \ref{prop:combined}} \label{ss:prop3proof}

\begin{proof} 
Let $\theta_0$ be the true value of $\theta$. Denote $p^{L+}(\Dobs_1, \theta_0),\ldots , p^{L+}(\Dobs_M, \theta_0)$ by $\widetilde{U}_1,\cdots, \widetilde{U}_M$, which are independent random variables that stochastically dominate the Uniform[0,1] random variable by part (a) of Proposition \ref{prop:FRTlowerCD}.

From (\ref{eq:comb2}), it follows that for any $0<\alpha<1$,  
\begin{align*}
   p_c^{L+}(\theta_0) &= P\left[G_c\left(g_c\left(p^{L+}(\Dobs_1, \theta_0),\ldots , p^{L+}(\Dobs_M, \theta_0)\right)\right)\leq \alpha\right]
    \\&=P\left[G_c\left(g_c\left(\widetilde{U}_1,\ldots, \widetilde{U}_M\right)\right)\leq \alpha\right]
    \\&= E\left[\mathbbm{1}_{\left(G_c\left(g_c(\widetilde{U}_1,\ldots, \widetilde{U}_M)\right)\leq \alpha\right)}\right]
    \\&= E\left[E\left[\mathbbm{1}_{\left(G_c\left(g_c(\widetilde{U}_1,\ldots, \widetilde{U}_M)\right)\leq \alpha\right)}\Big|\widetilde{U}_1,\ldots, \widetilde{U}_{M-1} \right]\right]
    \\&= E\left[P\left(\widetilde{U}_M\leq h_M^{-1}\left(G_c^{-1}(\alpha)\right)\big|\widetilde{U}_1,\ldots, \widetilde{U}_{M-1} \right)\right], \quad \text{where}\ h_M(t)=g_c(\widetilde{U}_1,\ldots, \widetilde{U}_{M-1},t).
    \\&\leq E\left[h_M^{-1}\left(G_c^{-1}(\alpha)\right)\right],\quad \text{since}\ P(\widetilde{U}_{M}\leq \epsilon)\leq \epsilon.
    \\ &= E\left[P\left({U}_M^\prime\leq h_M^{-1}\left(G_c^{-1}(\alpha)\right)\big|\widetilde{U}_1,\ldots, \widetilde{U}_{M-1} \right)\right], \ \text{where}\ {U}_M^\prime\sim \text{Unif[0,1], independent of} \ \widetilde{U}_1,\ldots, \widetilde{U}_{M}. \\
    &= E\left[E\left[\mathbbm{1}_{\left(G_c\left(g_c(\widetilde{U}_1,\ldots, \widetilde{U}_{M-1}, U_M^\prime)\right)\leq \alpha\right)}\Big|\widetilde{U}_1,\ldots, \widetilde{U}_{M-1} \right]\right]
    \\&= P\left[G_c\left(g_c\left(\widetilde{U}_1,\ldots, \widetilde{U}_{M-1},U_M^\prime\right)\right)\leq \alpha\right]
    \\&\leq P\left[G_c\left(g_c\left({U}_1^\prime,\ldots, {U}_{M-1}^\prime, {U}_M^\prime\right)\right)\leq \alpha\right], \ \text{repeating the previous step successively for} \ \widetilde{U}_{M-1}, \ldots, \widetilde{U}_1. 
    \\&=\alpha.
\end{align*}
Thus, $p_c^{L+}(\theta)$ is a valid lower CD. A similar argument can be applied to show that $p_c^{U+}(\theta)$ is also a valid upper CD.
\end{proof}

\subsubsection*{Proof of Theorem \ref{thm:upper_bound}} \label{ss:thm3proof}

\noindent We first introduce/recall the following notations: Let ${\Wrep_1,\ldots, \Wrep_K}$ be independent assignment vectors from the assignment mechanism and ${\Wrepprime_1,\ldots, \Wrepprime_K}$ be an independent copy. For $\theta\in \mathbb{R}$, let $\Yimptheta$ be the imputed potential outcome matrix under sharp null hypothesis $H_0^\theta: Y_i(1)-Y_i(0)=\theta,  \ \forall i=1,\cdots K$.  

Define 
\begin{equation}
\theta(\Wrep_j)=\mathbbm{1}\left\{T(\Yimptheta,\Wrep_j)\geq \Tobs \right\}\quad \text{for} \ j=1,2,\cdots K. \label{eq:indicator}
\end{equation}

To simplify notations, we denote the Monte-Carlo estimator $\widehat{p}_K^{L+}(\Dobs, \theta)$ of the $p$-value function $p^{L+}(\Dobs, \theta)$ based on assignment vectors $\Wrep_1,\ldots, \Wrep_K$ by $\delta_K$, which by (\ref{eq:indicator}) and (\ref{eq:p_est}) can be written as
\begin{equation}
\delta_K = \frac{1}{K}\sum_{j=1}^K \theta(\Wrep_j). \label{eq:deltaK}
\end{equation}

Similarly, for the independent copy ${\Wrepprime_1,\ldots, \Wrepprime_K}$, we can write
\begin{equation}
\delta_K^{\prime} = \frac{1}{K}\sum_{j=1}^K \theta(\Wrepprime_j). \label{eq:deltaKprime}
\end{equation}

We now state and prove three lemmas. Lemma \ref{lemma1} and Lemma \ref{lemma3} are useful in proving the theorem, and Lemma \ref{lemma2} helps establish Lemma \ref{lemma3}.

\begin{lemma} \label{lemma1}
Let $\epsilon_1,\cdots\epsilon_K$ be iid symmetric Bernoulli random variables, which are also independent of $\Wrep_1,\ldots, \Wrep_K$ and ${\Wrepprime_1,\ldots, \Wrepprime_K}$. Then,
$$\sup_\theta \left|\sum_{j=1}^K \epsilon_j \theta({\Wrep_j)} \right| \le \max_{1\leq j\leq K} \left|\sum_{i=1}^j \epsilon_i \right|,$$
where $\theta({\Wrep_j)}$ is defined in (\ref{eq:indicator}) and the underlying test statistic $T = T(\Yimptheta, \Wrep)$ is a non-decreasing and right continuous function of $\theta$ for fixed $\Wrep$.
\end{lemma}

\begin{proof}
For fixed $\Wrep_1,\ldots, \Wrep_K$, define
$$ \theta_j^* = \inf \{\theta: T(\Yimptheta,\Wrep_j)\geq \Tobs \}, \quad j=1, \ldots, K. $$
Because $T( \Yimptheta, \Wrep_j)$ is non-decreasing in $\theta$ and is right continuous, it follows that 
$$
\theta(\Wrep_j) = \left\{
        \begin{array}{ll}
            0, & \quad \theta < \theta_j^*, \\
            1, & \quad \theta \geq \theta_j^*,
        \end{array}
    \right.
$$

Without loss of generality, assume that $-\infty\leq\theta_1^*\leq \ldots \leq \theta_K^*\leq \infty$. Then  
$$
\sum_{j=1}^K \epsilon_j \theta(\Wrep_j) = \left\{
        \begin{array}{ll}
            0, & \quad \theta < \theta_1^*, \\
            \epsilon_1, & \quad \theta_1^*\leq \theta < \theta_2^*,\\
            \epsilon_1+\epsilon_2, & \quad \theta_2^*\leq \theta < \theta_3^*,\\
            \vdots\\
            \epsilon_1+\ldots+\epsilon_K, & \quad \theta \geq \theta_K^*.
        \end{array}
    \right.
$$
Consequently,
$$\sup_\theta \left|\sum_{j=1}^K \epsilon_j \theta({\Wrep_j)} \right|\leq \max_{1\leq j\leq K} \left|\sum_{i=1}^j \epsilon_i \right|.$$
\end{proof}

\begin{lemma} \label{lemma2}
Let $\epsilon_1,\cdots\epsilon_K$ be iid symmetric Bernoulli random variables. For any $c\in \mathbb{R}$, 
$$P\left(\max_{1\leq j\leq K} \left|\sum_{i=1}^j \epsilon_i \right|\geq c\right)\leq 2P\left( \left|\sum_{i=1}^K \epsilon_i \right|\geq c\right).$$
\end{lemma}

\begin{proof}
\noindent Since $\max_{1\leq j\leq K} \left|\sum_{i=1}^j \epsilon_i \right|$ and $\left|\sum_{i=1}^K \epsilon_i \right|$ can only take values $0, 1, \ldots, K$, it suffices to consider $c \in \{0, 1, \ldots, K\}$.
\begin{align*}
P\left(\max_{1\leq j\leq K} \left|\sum_{i=1}^j \epsilon_i \right|\geq c\right)& =P\left(\max_{1\leq j\leq K} \left|\sum_{i=1}^j \epsilon_i \right|\geq c, \left|\sum_{i=1}^K \epsilon_i \right|\geq c\right) + P\left(\max_{1\leq j\leq K} \left|\sum_{i=1}^j \epsilon_i \right|\geq c, \left|\sum_{i=1}^K \epsilon_i \right|< c\right)
\\&\leq P\left( \left|\sum_{i=1}^K \epsilon_i \right|\geq c\right)+ P\left(\max_{1\leq j\leq K} \left|\sum_{i=1}^j \epsilon_i \right|\geq c, \left|\sum_{i=1}^K \epsilon_i \right|< c\right)
\end{align*}
It suffices to prove that $$P\left(\max_{1\leq j\leq K} \left|\sum_{i=1}^j \epsilon_i \right|\geq c, \left|\sum_{i=1}^K \epsilon_i \right|< c\right)\leq P\left( \left|\sum_{i=1}^K \epsilon_i \right|\geq c\right).$$ For any possible choice of $(\epsilon_1,\epsilon_2,\cdots,\epsilon_K)$ satisfying $\max_{1\leq j\leq K} \left|\sum_{i=1}^j \epsilon_i \right|\geq c$ and $\left|\sum_{i=1}^K \epsilon_i \right|< c$, by reflection principle (reflect around the last intersection point with line $y=c$ or $y=-c$ as shown in Figure \ref{fig:reflection}), there is a unique $(\epsilon_1^\prime,\epsilon_2^\prime,\cdots,\epsilon_K^\prime)$ satisfying $\left|\sum_{i=1}^K \epsilon_i^\prime \right|\geq c$. Besides, the probability of every choice of $(\epsilon_1,\epsilon_2,\cdots,\epsilon_K)$ is $1/2^K$. This yields the desired inequality.
 \begin{figure}[ht]
	\centering 
	\caption{Reflection principle} \label{fig:reflection}
	\includegraphics[width=8cm]{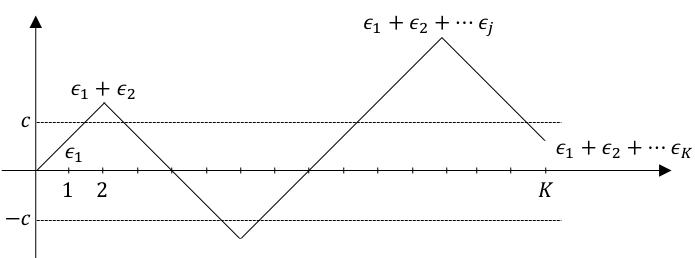}\\\vspace{1cm}
	\includegraphics[width=8cm]{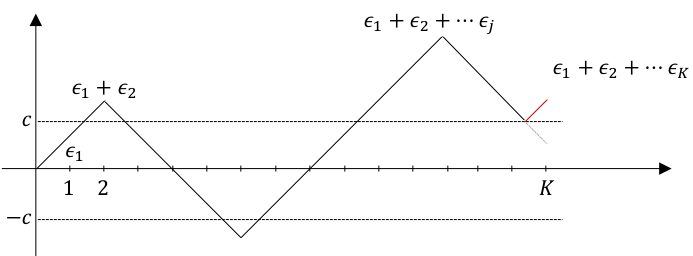}
    \end{figure}
\end{proof}
		
\begin{lemma} \label{lemma3}
Let $\epsilon_1,\cdots\epsilon_K$ be defined as in Lemma \ref{lemma2}. Then, for any $t>0$
$$E\left[ e^{2t\max_{1\leq j\leq K} \left|\frac{1}{K}\sum_{i=1}^j \epsilon_i \right|}\right]\leq 2E\left[ e^{2t \left|\frac{1}{K}\sum_{i=1}^K \epsilon_i \right|}\right].$$
\end{lemma}		

\begin{proof}
The proof is a straightforward application of Lemma \ref{lemma2} and the well known result $EX=\int_0^\infty P(X\geq t) dt$ for any nonnegative random variable $X$.
\end{proof}

We now return to the proof of the main theorem. For $t > 0$, we have,
\begin{align*}
& P\left(\sup_\theta |\delta_K-p^{L+}(\Dobs,\theta)|> \epsilon\right)\\
&\leq e^{-t\epsilon}E\left(e^{t\sup_\theta |\delta_K-p^{L+}(\Dobs,\theta)|}\right)\quad \text{by Markov inequality}
\\&=e^{-t\epsilon}E\left(\sup_\theta e^{t |\delta_K-p^{L+}(\Dobs,\theta)|}\right)
\\&=e^{-t\epsilon}E\left(\sup_\theta e^{t |\delta_K-E\delta_K^\prime|}\right)
\\&=e^{-t\epsilon}E\left(\sup_\theta e^{t \big|E\{\delta_K-\delta_K^\prime |\delta_K\}\big|}\right)
\\&\leq e^{-t\epsilon}E\left(\sup_\theta E\left( e^{t |\delta_K-\delta_K^\prime |}\Big|\delta_K\right)\right) \ \text{by Jensen's inequality}
\\&\leq e^{-t\epsilon}E\left( E\left(\sup_\theta e^{t |\delta_K-\delta_K^\prime |}\Big|\delta_K\right)\right) 
\\&= e^{-t\epsilon}E\left(\sup_\theta e^{t |\delta_K-\delta_K^\prime|}\right)
\\&=e^{-t\epsilon}E\left[\sup_\theta e^{t \Big|\frac{1}{K}\sum_{j=1}^K \left(\theta(\Wrep_j)-\theta(\Wrepprime_j)\right)\Big|}\right] \quad \text{by} \ (\ref{eq:deltaK}) \ \text{and} \ (\ref{eq:deltaKprime})
\\&=e^{-t\epsilon}E\left[\sup_\theta e^{t \Big|\frac{1}{K}\sum_{j=1}^K \epsilon_j\left(\theta(\Wrep_j)-\theta(\Wrepprime_j)\right)\Big|}\right] \quad \text{where} \ \epsilon_j \overset{iid}{\sim} \ \text{symmetric Bernoulli}
\\ &\leq e^{-t\epsilon}E\left[\sup_\theta e^{\frac{1}{2} 2t |\frac{1}{K}\sum_{j=1}^K \epsilon_j \theta(\Wrep_j)|+\frac{1}{2} 2t |\frac{1}{K}\sum_{j=1}^K \epsilon_j \theta(\Wrepprime_j)|}\right]\quad \text{by the triangle inequality}
\\&\leq e^{-t\epsilon}E\left[\sup_\theta e^{2t |\frac{1}{K}\sum_{j=1}^K \epsilon_j \theta(\Wrep_j)|}\right] \quad \text{by Jensen's inequality and properties of the supremum}
\\& \leq e^{-t\epsilon}E\left[ e^{2t\sup_\theta |\frac{1}{K}\sum_{j=1}^K \epsilon_j \theta(\Wrep_j)|}\right]
\\&\leq e^{-t\epsilon}E\left[ e^{2t\max_{1\leq j\leq K} |\frac{1}{K}\sum_{i=1}^j \epsilon_i |}\right]\quad  \text{by Lemma \ref{lemma1}} 
\\&\leq 2e^{-t\epsilon}E\left[ e^{2t |\frac{1}{K}\sum_{i=1}^K \epsilon_i |}\right] \quad \text{by Lemma \ref{lemma3}}
\\&\leq 2e^{-t\epsilon}E\left[ e^{\frac{2t}{K}\sum_{i=1}^K \epsilon_i }+e^{\frac{-2t}{K}\sum_{i=1}^K \epsilon_i }\right]
\\&=4e^{-t\epsilon}E\left[ e^{\frac{2t}{K}\sum_{i=1}^K \epsilon_i }\right]
\\&\leq 4e^{-t\epsilon}e^{\frac{2t^2}{K}} \quad \text{by Hoeffding’s inequality}.
\end{align*}
The desired result is obtained by taking $t=\frac{K\epsilon}{4}$ in the above inequality.

\end{document}